\documentclass[fleqn]{article}
\usepackage{amsmath, amsthm, amssymb}
\usepackage{physics}
\usepackage{xcolor}
\usepackage{graphicx}
\usepackage{authblk}
\usepackage[hidelinks]{hyperref}

\newcommand{\R}{\mathbb{R}}
\newcommand{\C}{\mathbb{C}}

\newtheorem{theorem}{Theorem}[section]
\newtheorem{lemma}[theorem]{Lemma}
\newtheorem{corollary}[theorem]{Corollary}
\newtheorem{remark}[theorem]{Remark}

\newcommand{\bb}[1]{\left(#1\right)} 
\newcommand{\seq}[1]{\left\{#1\right\}} 
\renewcommand{\sb}[1]{\left[#1\right]} 

\newcommand{\cb}[1]{\left\{#1\right\}}
\newcommand{\D}{\mathrm{d}}
\newcommand{\I}{\mathrm{i}}
\newcommand{\Iuc}{I_{\mathrm{uc}}}
\newcommand{\Ibz}{I_{\mathrm{bz}}}

\newcommand{\figpath}{./}

\begin{document}

\begin{center}
   \begin{minipage}[t]{6.0in}
		A standard task in solid state physics and quantum chemistry is the computation of localized molecular
orbitals known as Wannier functions. In this manuscript, we propose a new procedure for computing Wannier functions in one-dimensional crystalline materials. Our approach proceeds by first performing parallel transport of the Bloch functions using numerical integration. Then a simple analytically computable correction is introduced to
yield the optimally localized Wannier function. The resulting scheme is rapidly convergent and is proven to yield real-valued Wannier functions that achieve global optimality. The analysis in this manuscript can also be viewed as a proof of the existence of exponentially localized Wannier functions in one dimension. We illustrate the performance of the scheme by a number of numerical experiments.
	\thispagestyle{empty}

  \vspace{ -100.0in}

  \end{minipage}
\end{center}

\vspace{ 3.60in}

\begin{center}
  \begin{minipage}[t]{4.4in}
    \begin{center}

\textbf{A highly accurate procedure for computing globally optimal Wannier functions in one-dimensional crystalline insulators } \\

  \vspace{ 0.50in}

Abinand Gopal$\mbox{}^{\dagger}$,
Hanwen Zhang$\mbox{}^{\ddagger }$  \\
              \today

    \end{center}
  \vspace{ -100.0in}
  \end{minipage}
\end{center}

\vspace{ 2.00in}

\vfill
\noindent
$\mbox{}^{\dagger}$ Department of Mathematics, UC Davis, Davis, CA 95616\\
\noindent
$\mbox{}^{\ddagger}$ Department of Applied and Computational Mathematics, Yale University, New Haven, CT 06511

\vspace{2mm}





\vfill
\eject

\tableofcontents

\section{Introduction}
The eigenfunctions of the Schr\"odinger equation form a natural basis for
electronic structure calculations.
These eigenfunctions provide a complete description of the quantum states of
a system and are the foundation for analyzing spectral properties, energies,
and response to external perturbations. 
For atoms and molecules in isolation, the eigenfunctions are spatially
localized, reflecting the confinement of electrons to specific regions of
space, and they lend themselves to intuitive chemical and physical
interpretations.
In crystalline materials, on the other hand, the situation changes
fundamentally. 
The periodicity of the potential leads, by Bloch’s theorem, to eigenfunctions (known in this case as Bloch functions) that are not localized in real space but instead exhibit quasiperiodicity.
The delocalized nature of the eigenfunctions makes them less useful in electronic structure calculations. 
More useful in this setting are Wannier functions which are localized
molecular orbitals in crystalline materials. 
These then can be used for the tight-binding approximation and to understand
the conductivity and polarization properties of certain materials (cf.\
\cite{vanderbilt2018berry}).
As such, the computation of Wannier functions is of interest in the
development of new materials, such as semiconductors, solar cells, and
topological insulators.

Formally, Wannier functions are obtained as the Fourier transform of Bloch functions with respect to crystal momentum. 
However, the Bloch functions are defined only up to a momentum-dependent phase (or, more generally, unitary transformation), often referred to in physics as a ``gauge.''
Different choices of gauge yield different Wannier functions, with the most relevant for applications being the optimally localized Wannier functions. 
Mathematically, the construction amounts to solving a family of
parameter-dependent eigenvalue problems while ensuring that the eigenvectors
vary smoothly with respect to the momentum parameter.

The most common approach is to first compute the Bloch functions and then choose the gauge to  
minimize the variance of the resulting Wannier function 
(known as the Marzari--Vanderbilt functional) using some form of gradient descent with high quality initial guesses computed by density-matrix based methods \cite{damle2019variational,damle2018disentanglement,pizzi2020wannier90}. {Unique to the one-dimensional case, it is also known that the globally optimally localized Wannier functions are the eigenfunctions of the projected position operator\,\cite{nenciu1998existence,vanderbilt2018berry}.}

In this manuscript, we propose an alternative approach that requires no
iterative optimization and is {generalizable to higher dimensions}.
We solve a parallel transport problem via numerical integration
to determine the Bloch functions and band functions.
{We then show that a phase correction yields real-valued Wannier functions that achieve global optimality, thereby eliminating the need for iterative optimization or for solving the eigenvalue problem of the projected position operator. 
We also provide complete analysis for the construction, establishing for the first time a constructive proof for the existence of exponentially localized Wannier function in one dimension possessing global optimality via parallel transport.} 
The analysis in this paper for exponential localization is based on Kato's analytic perturbation theory and is related to the work \cite{cornean2016construction}. 
The identification of the global optimal solution is a simple consequence of the analysis.

The algorithm in this paper is closely related to the one dimensional case in 
\cite{cances2017robust,marzari1997maximally,vanderbilt2018berry}, which are also based on parallel transport. 
However, our approach differs from these existing works in three crucial
aspects.
First, we provide a complete numerical procedure that starts with the
Schr\"odinger equation and ends with the evaluation of the Wannier function, where {we solve the parallel transport equation exactly up to discretization errors only. }
Second, this convergence of our approach is of arbitrarily high-order with the order determined solely by choice of integrator. 
This is in contrast to the approach of
\cite{cances2017robust,marzari1997maximally,vanderbilt2018berry} which has a parallel transport procedure that is fundamentally limited to low order. 
Finally and most significantly, we are able to directly compute a provably 
globally optimal Wannier function without any type of iterative
optimization.  

{Furthermore, we point out that our analysis shows the parallel transport scheme in one dimension in \cite{marzari1997maximally,vanderbilt2018berry}  (see Remark\,\ref{rmk:twist}), although only a second-order approximation, also yields exponentially localized Wannier functions when combined with the same correction procedure in this paper. To the best of our knowledge, this has long been a heuristic among physicists \cite{marzari2012maximally, vanderbilt2018berry}, but has never been rigorously proven. The resulting Wannier functions from this scheme, however, differ slightly from the globally optimal ones presented in this work, with the discrepancy arising from the approximation in the parallel transport step. This observation is particularly relevant in settings where it is advantageous to first obtain the Bloch functions, such as in density functional theory.  To avoid digression, we will report a detailed analysis of the above points at a later date.}

This is the first in a series of papers;
the extension of our results to single band matrix models in two dimensions
can be found in the preprint \cite{zhang2025constructing}.  In subsequent papers we will consider the multiband case both in one dimension
 and higher dimensions. We will briefly mention the challenges for these extensions in Section\,\ref{s:conclusions}.

We now outline the remainder of this manuscript.
In Section \ref{s:prelims}, we review some background material on Wannier
functions.
This is followed by Section \ref{s:apparatus} where we develop the analytic
apparatus for constructing globally optimal Wannier functions. 
After this are Sections \ref{s:procedure} and \ref{s:detailed} where we
describe our numerical procedure.
Our numerical procedure is illustrated in Section \ref{s:numerics} which contains numerical experiments. 
We provide our conclusions and outline future work in Section
\ref{s:conclusions}. 

\section{Physical and Mathematical Preliminaries}
\label{s:prelims}

\subsection{Notation}
We let $\mathbb{N}$ denote the positive integers.
We let $\delta_{mn}$ denote the Kronecker delta function, defined by 
\begin{equation}
\delta_{mn} = \begin{cases} 1 & m = n \\ 0 & m \neq n \end{cases}.
\end{equation}

We let $\ell^2(\mathbb{Z})$ denote the Hilbert space for complex sequences
$\mathbf{x}=\seq{x_i}_{i=-\infty}^{\infty}$ such that
$\sum_{i=-\infty}^{\infty} \abs{x_i}^2 < \infty$  with the standard inner product. 
We denote the right shift operator in $\ell^2(\mathbb{Z})$ by $R$, where
\begin{equation}
	R(\mathbf{x}) = \seq{x_{i-1}}_{i=-\infty}^{\infty}.
	\label{eq:rshift}
\end{equation}
For a vector $\mathbf{x} \in \ell^2(\mathbb{Z})$ or $\mathbb{C}^N $, we denote by
$\mathbf{x}^*$ the conjugate transpose and by
$\| \mathbf{x} \|$ the standard 2-norm, $\| \mathbf{x} \| =
\sqrt{\mathbf{x}^* \mathbf{x}}$.

Given sets $A \subset \mathbb{C}$ and $B \subset \mathbb{C}$, we write $f
\in C(A;B)$ if $f: A \to B$ is a continuous function.
Similarly, we write
$f \in C^n(A; B)$ if $f$ is a $n$-times continuously differentiable. 
For an $\mathbf{A}$ as an operator in $\ell^2(\mathbb{Z})$ or a matrix in $ \mathbb{C}^{M \times M}$, we let
$\mathbf{A}^\dagger$ denote its Moore--Penrose pseudoinverse. 

We choose the principal branch of the $\log$ function so that we  assume all real phases $\varphi$ in $e^{\I \varphi}$ take value in $(-\pi,\pi]$.

\subsection{Bloch functions \label{sec:bloch}}
In this section, we introduce Bloch functions in one dimension. 
These are standard facts in the theory of bands and can all be found in \cite{kohn1959analytic}, for example.

We first define a periodic lattice with a period (lattice constant) $a$
\begin{equation}
\Lambda = \left\{ n a : n \in \mathbb{Z}\right\}.
\label{eq:lattice}
\end{equation}
We refer to the interval $ [-a/2,a/2)$ as the primitive unit cell:
\begin{equation}
	\Iuc = [-a/2,a/2).
\end{equation}
It is also convenient to define 
\begin{equation}
\Omega=\frac{2\pi}{a}
\label{eq:om22}
\end{equation}
 as the period for the reciprocal lattice
\begin{equation}
\Lambda^* = \left\{ n \Omega : n \in \mathbb{Z}\right\}.
\label{eq:rec}
\end{equation}
The interval $[-\Omega/2 ,\Omega/2)$ is usually called the first Brillouin zone (BZ), and we denote the interval by $\Ibz:$
\begin{equation}
	\Ibz = [-\Omega/2 ,\Omega/2).
	\label{eq:ibz22}
\end{equation}
We assume that $V : \mathbb{R} \to \mathbb{R}$ is a piecewise continuous function that
has a period $a$:
\begin{equation}
V(x+y) = V(x), \quad x \in \mathbb{R},\, y \in \Lambda.
\label{eq:potper}
\end{equation}

For any $k \in \mathbb{R}$, the relevant Schr\"odinger equation is given by
the eigenvalue problem
\begin{equation}
 -\frac{{\rm d}^2 \psi_k^{(j)}}{{\rm d}x^2} + V(x) \psi_k^{(j)} = E_k^{(j)} \psi_k^{(j)},
 \label{eq:schro}
\end{equation}
subject to the boundary condition
\begin{equation}
\psi_k^{(j)}(x+y) = e^{\I ky} \psi_k^{(j)}(x), \quad y \in \Lambda.
\label{eq:bc}
\end{equation}
The value $k$ is called the quasimomentum, $\{\psi_k^{(j)}\}_{j=1}^\infty$
in \eqref{eq:schro} are called the Bloch functions, and
$\{E_k^{(j)}\}_{j=1}^\infty$ are called the energy of the states represented by $\{\psi_k^{(j)}\}_{j=1}^\infty$.


Bloch's theorem states two simplifications of the above problem. First, it is sufficient to only consider $k\in \Ibz$, the first BZ. Second, the Bloch functions are of the form
\begin{equation}
	\psi_k^{(j)}(x) = e^{\I kx} u^{(j)}_k(x),
	\label{eq:ufun}
\end{equation}
for some function $u^{(j)}_k$ where 
\begin{equation}
	u^{(j)}_k(x + y) = u^{(j)}_k(x), \quad y \in \Lambda. 
	\label{eq:ubc}
\end{equation}
Substituting (\ref{eq:ufun}) into (\ref{eq:schro}), shows
 $u^{(j)}_k$ satisfies
\begin{equation}
	 \sb{-\bb{\frac{\D}{\D x} + \I k}^2 + V }u_k^{(j)} = E_k^{(j)} u_k^{(j)},
	 \label{eq:ueq}
\end{equation}
subject to the boundary condition in (\ref{eq:ubc}). 
It is convenient to denote the operator in (\ref{eq:ueq}) densely defined
in $L^2(I_{\rm uc})$ by
\begin{equation}
	H(k) = -\bb{\frac{\D}{\D x} + \I k}^2  + V.
	\label{eq:hamilt}
\end{equation}
For any real piecewise continuous function $V$, $H(k)$ is self-adjoint for real
$k$ and the operator has the symmetry
\begin{equation}
	H(k) = \overline{H}(-k)
	\label{eq:trsym}
\end{equation}
for real $k$. 
This symmetry is referred to as time-reversal symmetry in physics literature. 

The following result in \cite{reed1978analysis} asserts that, if 
$V$ is piecewise continuous, then the operator associated with \eqref{eq:ueq} is analytic in 
$k$ in the standard sense (cf. \cite{kato2013perturbation,reed1978analysis}); that is, its resolvent depends analytically on $k$ as a bounded operator. Furthermore, the operator has purely discrete spectrum.
\begin{lemma}
	Suppose the potential $V$ in (\ref{eq:ueq}) is piecewise continuous in $\Iuc$.  Then the following holds for any $k\in \Ibz$: \\
1. The operator $H(k)$ in (\ref{eq:hamilt}) is analytic in $k$.\\
2. The operator $H(k)$ in (\ref{eq:hamilt}) has purely discrete spectrum.
\end{lemma}
This justifies the labeling of energies and Bloch functions in (\ref{eq:schro}) and (\ref{eq:ueq}). 

\begin{remark}
The eigenfunctions $\{\psi_k^{(j)}\}_{j=1}^\infty$ are only defined up to a
nonzero constant scaling.
We can assume that these modulus of these constants are chosen such that  
\begin{equation}
\int_{-a/2}^{a/2} |\psi_k^{(j)}(x)|^2\,{\rm d}x= 
\int_{-a/2}^{a/2} |u_k^{(j)}(x)|^2\,{\rm d}x = 1, \quad k \in \mathbb{R},\, j=1,2,\hdots.
\label{eq:unorm}
\end{equation}
\end{remark}

There are different conventions for extending $k$ outside of the first BZ. 
In the context of Wannier functions, it is natural to use the so-called periodic zone scheme, where the Bloch functions and eigenvalues are copies of those in the first BZ \cite{vanderbilt2018berry}. 
In other words, we can extend the Bloch functions and eigenvalues to be periodic functions in $k$ with period $\Omega$. More explicitly, we have
\begin{equation}
	E_{k+g}^{(j)} = E_k^{(j)}
		\label{eq:ep}
\end{equation}
and
\begin{equation}
	\psi_{k+g}^{(j)} = \psi_{k}^{(j)}
	\label{eq:psip}
\end{equation}
for any $k\in I_{\rm bz}, g\in\Lambda^*$ and $j=1,2,\ldots$. Combining (\ref{eq:psip}) and (\ref{eq:ufun}) shows that
\begin{equation}
	u_{k+g}^{(j)}(x) = e^{-\I gx}u_{k}^{(j)}(x).
	\label{eq:ubck}
\end{equation}

\subsection{Wannier functions \label{sec:wannier}}
Due to the periodic condition \eqref{eq:psip}, the Wannier functions are defined as the Fourier coefficients of $\{\psi_k^{(j)}(x)\}_{j=1}^{\infty}$:
\begin{equation}
W_n^{(j)}(x) = \frac{1}{\Omega}\int_{-\Omega/2}^{\Omega/2} e^{-\I nak }\psi_k^{(j)}(x)\,{\rm d}k, \quad j = 1,2,\hdots.
\label{eq:wanniern}
\end{equation}
The Wannier functions with index $n \neq 0$ are copies of
$W_0^{(j)}$ shifted to be centered around the $n$th lattice site \cite{vanderbilt2018berry}. 
More explicitly, we have
\begin{equation}
W_n^{(j)}(x) = W_0^{(j)}(x-na), \quad j=1,2,\hdots.
\label{eq:wancopy}
\end{equation}
Hence, we only need to consider the case where $n=0$:
\begin{equation}
W_0^{(j)}(x) = \frac{1}{\Omega}\int_{-\Omega/2}^{\Omega/2} \psi_k^{(j)}(x)\,{\rm d}k = \frac{1}{\Omega}\int_{-\Omega/2}^{\Omega/2} e^{\I xk}u_k^{(j)}(x)\,{\rm d}k, \quad j = 1,2,\hdots.
\label{eq:wannier}
\end{equation}

Wannier functions are highly non-unique since the Bloch functions in (\ref{eq:schro}) with the normalization condition in (\ref{eq:unorm}) are only determined up to a complex number of unit modulus. More explicitly, if any $\psi_k^{(j)}$ satisfies (\ref{eq:schro}), so will any $e^{-\I\varphi(k)}\psi_k^{(j)}$ with a real $\varphi(k)$ for $k\in\Ibz$ without affecting the normalization (\ref{eq:unorm}). In this paper, we refer to numbers of the form $e^{-\I\varphi}$ as phase factors.  For both physical and computational purposes, it is highly desirable to choose $\psi_k^{(j)}$ so that the Wannier functions are well-localized.

It was shown in \cite{kohn1959analytic}, in the case of a symmetric
potential, and in \cite{des1964analytical}, more generally, that
$\{\psi_k^{(j)}\}_{j=1}^\infty$ can be chosen such that the
$\{W_0^{(j)}\}_{j=1}^\infty$ are exponentially localized.
In other words, there exist constants $C > 0$ and $D > 0$ such that
\begin{equation}
|W_0^{(j)}(x)| \leq C e^{-D |x|}, \quad x \in \mathbb{R}.
\label{eq:exploc}
\end{equation}

We also define the moment functions 
\begin{equation}
\langle x \rangle = \int_{-\infty}^\infty x |W_0^{(j)}(x)|^2 \,{\rm d}x,
\label{eq:mom1}
\end{equation}
and
\begin{equation}
\langle x^2 \rangle = \int_{-\infty}^\infty x^2 |W_0^{(j)}(x)|^2\,{\rm d}x.
\label{eq:mom2}
\end{equation}
The variance defined by $\langle x^2 \rangle - \langle x \rangle^2$ is a standard localization measure for Wannier functions\,\cite{vanderbilt2018berry}.


\subsection{Analyticity of eigenvalues and eigenvectors \label{sec:katoan}}
We next introduce standard results in analytic perturbation theory for linear operators in \cite{kato2013perturbation}. This includes smoothness results for eigenvalues/vectors of an analytic family of operators, and formulas for their derivatives. 

Suppose we have a family of matrices $T(z) \in \C^{n\times n}$, where $T(z)$ is
analytic for $z$ in a domain $D\subset \C$ intersecting the real axis. 
We call an analytic family of matrices $T(z)$ self-adjoint if it satisfies
$T^*(z) = T(\bar{z})$ for $z\in D$. 
The following theorem is Theorem 6.1 in \cite{kato2013perturbation}.

\begin{theorem}
	If an analytic family $T(z)\in \C^{n\times n}$ is self-adjoint, its eigenvalues $\cb{\lambda_i(z)}_{i=1}^{n}$ and eigenprojectors $\cb{P_i(z)}_{i=1}^{n}$ are analytic for real $z\in D$.
	\label{thm:katoproj}
\end{theorem}

It is worth noting that the above theorem is still applicable when the eigenvalues become degenerate. 
In this case, the eigenprojectors at degenerate points must be viewed as limits of those approaching the degenerate points. We refer the reader to Theorem 1.10 and Remark 1.11 in \cite{kato2013perturbation} for details.


\begin{remark}
Theorem \ref{thm:katoproj} concerns only the eigenprojectors, not the eigenvectors; this is because for any an eigenvector $\mathbf{v}_i(z)$ of $T(z)$, one can multiply the vector by a $z$-dependent phase factor $e^{-\I\varphi(z)}$ ($\varphi(z)$ is real for real $z$) while the new vector $e^{-\I\varphi(z)}\mathbf{v}_i(z)$ is still an eigenvector. The phase factor could change the smoothness of $\mathbf{v}_i(z)$ while the projector $P_i(z) = \mathbf{v}_i(z)\mathbf{v}^*_i(z)$ is unchanged for real $z$.
\label{rmk:phase}
\end{remark}

The following theorem shows that it is possible to find a family of eigenvectors $\mathbf{v}_i(z)$ that are analytic for real $z \in D$. 
It is a special case of the construction in Section 2.6.2 in
\cite{kato2013perturbation}. 
\begin{theorem}
	Suppose that the analytic matrix family $T(z)$ is self-adjoint in $D$, and that $\mathbf{v}_0$ is an eigenvector of $T(z_0)$ with a nondegenerate eigenvalue $\lambda_0$ for some real $z_0$ in $D$. 
Let $\lambda(z)$ and $P(z)$ be analytic families of eigenvalues and eigenprojectors in $D$, the existence of which follows from Theorem \ref{thm:katoproj}.
Consider the the initial value problem 
\begin{equation}
	\frac{\D}{\D z}\mathbf{v}(z) = Q(z)\mathbf{v}(z), \quad \mbox{with the initial condition} \quad \mathbf{v}(z_0) = \mathbf{v}_0,
	\label{eq:katoode}
\end{equation}
where $Q= \frac{\D P}{\D z} P - P\frac{\D P}{\D z}$.
The solution $\mathbf{v}(z)$ to (\ref{eq:katoode}) is analytic and satisfies
$T(z) \mathbf{v}(z) = \lambda(z)\mathbf{v}(z)$ for real $z$ in $D$.
\label{thm:katovec}
\end{theorem}


While the above results only apply to self-adjoint analytic family of finite-dimensional matrices $T(z)$ in $D$, they straightforwardly generalize to a class of infinite-dimensional operators in Hilbert space. 
The definition for the analytic family $T(z)$ being self-adjoint in a domain $D$ intersecting the real axes is unchanged. 
To extend Theorems \ref{thm:katoproj} and \ref{thm:katovec} to the operator case Kato defined a class of operators where a subset of their spectra can be separated by a simple close curve in the complex plane from the rest.

Kato showed that, for a self-adjoint analytic family $T(z)$ possessing a finite set of isolated spectrum that can be separated for $z\in D$, the finite-dimensional version Theorems \ref{thm:katoproj} and \ref{thm:katovec} can be applied to the finite-dimensional subspace corresponding to the separated eigenvalues. 
We summarize this extension in the following theorem and refer the reader to Sections 7.1.3 and 7.3.1 in \cite{kato2013perturbation} for details.
\begin{theorem}
	Suppose that $T(z)$ is an analytic family of operators in Hilbert space for $z\in D$. Suppose further that $T(z)$ is self-adjoint, defined analogously as the matrix case in Theorem \ref{thm:katoproj}. Theorems \ref{thm:katoproj} and \ref{thm:katovec} also apply to any finite systems of eigenvalues, eigenprojectors and eigenvectors of $T(z)$, where the eigenvalues can be separated from the rest of the spectrum of $T(z)$.
	\label{thm:katoop}
\end{theorem}
In this paper, we deal with operators with discrete spectra, where any eigenvalues of interests can be separated from the rest. 
Theorem \ref{thm:katoop} is applicable to the eigensubspaces of any such
operator.

\subsection{Derivatives of eigenvalues and eigenvectors}
Next, we provide expressions for the derivatives of eigenvalues and eigenvectors. 
For simplicity, we only state the result for the Hilbert space $\ell^2(\mathbb{Z})$ and operators with a purely discrete spectrum, which will suffice for our purpose. 
These formulas can be found in Section 8.2.3 in \cite{kato2013perturbation}.
The formula \eqref{eq:evecpert} is a consequence of the derivative of the
projection $\frac{\D P}{\D z}$ substituted into (\ref{eq:katoode}).
\begin{theorem}
	Suppose that $T(z)$ is a self-adjoint analytic family of operators in $\ell^2(\mathbb{Z})$ for $z\in D$ and the operators have a purely discrete spectrum. Suppose further that $z_0 \in D$ is real and $\lambda(z_0)$ is a nondegenerate eigenvalue of $T(z_0)$ with the eigenvector $\mathbf{v}$, i.e. $T(z_0)\mathbf{v}(z_0) = \lambda(z_0)\mathbf{v}(z_0)$. The derivative of $\lambda$ at $z_0$ is given by
	\begin{equation}
		\frac{\D \lambda(z_0)}{\D z} = \mathbf{v}(z_0)^* \frac{\D T(z_0)}{\D z} \mathbf{v}(z_0),
				\label{eq:evalpert}
	\end{equation}
and the derivative of $\mathbf{v}$ at $z_0$ is given by
	\begin{equation}
		\frac{\D \mathbf{v}(z_0)}{\D z} = - (T(z_0) - \lambda(z_0))^\dagger \frac{\D T(z_0)}{\D z} \mathbf{v}(z_0),
		\label{eq:evecpert}
	\end{equation}
	where $(T(z_0) - \lambda(z_0))^\dagger$ is the pseudoinverse of $T(z_0) - \lambda(z_0)$ ignoring the eigenspace of $\lambda(z_0)$.
\end{theorem}
We note that in obtaining (\ref{eq:evecpert}) from (\ref{eq:katoode}) we also used $P\frac{\D P}{\D z}P = 0$,  $P\mathbf{v}=\mathbf{v}$ for $\mathbf{v}$ in the range of $P$. 
We observe that, by the definition of pseudoinverse, the range of $(T - \lambda)^\dagger$ is orthogonal to that of $P$, so we have the following formula
\begin{equation}
	(T - \lambda)^\dagger P = P(T - \lambda)^\dagger = 0.	
	\label{eq:orthopv}
\end{equation}
We note that in the physics literature, the formula (\ref{eq:evecpert}) is usually written using the spectral decomposition of $T$ in the following form
	\begin{equation}
		\frac{\D \mathbf{v}(z_0)}{\D z} = \sum_{\substack{j \\ \lambda_j \ne \lambda }}\mathbf{v}_j(z_0) \frac{\mathbf{v}_j(z_0)^* \frac{\D T(z_0)}{\D z} \mathbf{v}(z_0)}{\lambda(z_0) - \lambda_j(z_0)},
		\label{eq:evecpertsvd}
	\end{equation}
	where $\lambda_j(z_0)$ and $\mathbf{v}_j(z_0)$ are the rest of the eigenvalues and eigenvectors of $T(z_0)$. 
This is the formula one would use to compute $\frac{\D \mathbf{v}}{\D z}$ using
the singular value decomposition of the operator $T(z_0) - \lambda(z_0)$.

\begin{remark}
\label{rmk:inv}
	We observe that \eqref{eq:evecpert} holds true if
$\frac{\D T(z_0)}{\D z}$ is replaced with $\frac{\D T(z_0)}{\D z} - \frac{\D \lambda(z_0)}{\D z}$. 
Since $(\frac{\D T(z_0)}{\D z} - \frac{\D \lambda(z_0)}{\D z}) \mathbf{v}(z_0)$ is orthogonal to $\mathbf{v}(z_0)$ (see (\ref{eq:orthopv})), we can compute $\frac{\D \mathbf{v}}{\D z}$ by
	\begin{equation}
		\frac{\D \mathbf{v}(z_0)}{\D z} = - (T(z_0) - \lambda(z_0) + \gamma P(z_0) )^{-1} \bb{\frac{\D T(z_0)}{\D z} - \frac{\D \lambda(z_0)}{\D z}} \mathbf{v}(z_0),
		\label{eq:evecpert2}
	\end{equation}
where $P(z_0)=\mathbf{v}(z_0)\mathbf{v}^*(z_0)$ is the projector and $\gamma$ is any reasonable constant such that the addition $\gamma P(z_0)$ makes the original operator invertible. 
As a result, the pseudoinverse is replaced by the inverse. Although (\ref{eq:evecpert2}), (\ref{eq:evecpert}) and (\ref{eq:evecpertsvd}) are all equivalent,  (\ref{eq:evecpert2}) is more attractive numerically in cases where the discretization of $T(z_0)$ is large.
\end{remark}

\section{Analytic apparatus}
\label{s:apparatus}

\subsection{Analysis of Bloch functions} 
In this section, we introduce results about analyticity and the Fourier analysis of the Bloch functions in Section \ref{sec:bloch}. In what follows, it will be convenient to view
$\{\psi_k^{(j)}\}_{j=1}^\infty$ as defined in \eqref{eq:schro} and $\{u_k^{(j)}\}_{j=1}^\infty$ defined in \eqref{eq:ueq} as functions
of both $x$ and $k$.
We will also view $E_k$ in \eqref{eq:schro} as a function of $k$. 

In this paper, we will restrict ourselves to a single band whose energy $E_k$ does not become degenerate for any $k$ (i.e., a single value
of $j$), so we will omit the superscript.
We thus rewrite (\ref{eq:ueq}), (\ref{eq:ubc}) and (\ref{eq:ubck}) as
\begin{equation}
H(k)u = -(\partial_{xx}  +\I k)^2 u(x,k) + V(x)  u(x,k) = E(k) u(x,k),
\label{eq:udiff}
\end{equation}
and
\begin{equation}
u(x,k)=u(x+y,k), \quad y \in \Lambda,\, k \in \Ibz
\label{eq:per3}
\end{equation}
\begin{equation}
	u(x,k+g)= e^{-\I gx}u(x,k), \quad g \in \Lambda^*, \,x \in \Iuc.
\label{eq:perk3}
\end{equation}

Since $H(k)$ is self-adjoint analytic family of operators and the eigenvalue $E(k)$ considered is isolated, Theorem \ref{thm:katoop} is applicable and we have the following lemma about the analyticity of $E(k)$ and $u$ in $k$.
\begin{lemma}
	The energy $E(k)$ in (\ref{eq:udiff}) is analytic in $k\in \Ibz$. The function $u(x,k)$ can be chosen to be analytic in $k\in \Ibz$ for $x\in \Iuc$.
	\label{lemma:euan}
\end{lemma}
\begin{remark}
We observe that Theorem \ref{thm:katoop} does not guarantee the function
$E(k)$ to be periodically analytic in $k\in\Ibz$ or $u(x,k)$ to satisfy (\ref{eq:perk3}), which will be
shown to be necessary for exponentially localized Wannier functions. 
If $E(k)$ were to become degenerate for some $k_0 \in \Ibz$, it could go to
another branch so that it does not return to $E(-\Omega/2)$ when it reaches
$E(\Omega/2)$. 
When $E$ is not degenerate, $E(k)$ will be periodically analytic in $k\in\Ibz$. 
For the eigenvector $u(x,k)$, as discussed in Remark {\ref{rmk:phase}},
whenever $u(x,k)$ is analytic in $k\in\Ibz$, so is any $e^{-\I \varphi(k)} u(x,k)$
with a real analytic $\varphi$, which need to satisfy (\ref{eq:perk3}).
\label{rmk:nonper}
\end{remark}

Exploiting \eqref{eq:per3}, we take the Fourier expansion of $u$ with
respect to the $x$ variable to get
\begin{equation}
u(x,k) = \sum_{m=-\infty}^\infty a_m(k) e^{\I m \Omega x},
\label{eq:ufour}
\end{equation}
where 
\begin{equation}
 \frac{1}{a}\int_{-a/2}^{a/2} u(x,k) e^{-\I m \Omega x} \D x= a_m(k), \quad m\in \mathbb{Z}.
 \label{eq:ufourint}
\end{equation}
By Parseval's theorem, the normalization condition for  $u(x,k)$ in (\ref{eq:unorm}) shows that
\begin{equation}
	\sum_{m=-\infty}^{\infty} \abs{a_m(k)}^2 = \frac{1}{a}.
	\label{eq:anorm}
\end{equation}
For simplicity, we assume for the rest of the paper that $u(x,k)$ as a function of $x$ is a periodic function in $
C^{1}(\Iuc;\mathbb{C})$  and
$\partial_x u(x,k)$ has bounded variation
for all $k\in \Ibz$, so that the Fourier coefficients $a_m(k)$ decays as $\abs{m}^{-2}$ for large $m$ \cite[Section 2.6]{zygmund2002trigonometric}. 
Thus, we have a positive constant $D$ independent of $k$, such that
\begin{equation}
|a_m( k)| \leq \frac{D}{|m|^2}, \quad \mbox{as $\abs{m} \rightarrow \infty$}\,.
\label{eq:decay}
\end{equation}
Hence the Fourier series in \eqref{eq:ufour} converges absolutely independent of $k$. By (\ref{eq:ufourint}), if $u(x,k)$ is chosen to be analytic for $k\in\Ibz$ by Lemma \ref{lemma:euan}, we observe that all the Fourier coefficients $a_m(k)$ are analytic for $k\in\Ibz$.
\begin{lemma}
	Suppose $u(x,k)$ is chosen to be analytic in $k\in\Ibz$. Its Fourier coefficients $a_m(k)$ in (\ref{eq:ufour}) are analytic in $k\in\Ibz$.
	\label{lemma:aman}
\end{lemma}

The following lemma provides a simple relation for the dependence of $k$ and
$m$ in the Fourier coefficients in \eqref{eq:ufour}. It is a well-known relation in the context of periodic zone scheme (see (\ref{eq:ep})-(\ref{eq:ubck})). It is a simple consequence of (\ref{eq:perk3}).
\begin{lemma}
There exists a function $\alpha :\mathbb{R} \to \mathbb{C}$
such that $a_m(k)$ in \eqref{eq:ufour} can be
expressed as
\begin{equation}
a_m(k) = \alpha(k+m \Omega), \quad k \in \Ibz,\, m \in \mathbb{Z}.
\label{eq:afun}
\end{equation}
\label{lemma:afun}
\end{lemma}
\begin{proof}
We let $n \in \mathbb{Z}$ and $g = n \Omega$. 
From \eqref{eq:perk3} and \eqref{eq:ufour} that 
\begin{align}
u(x,k+g) &= \sum_{m=-\infty}^\infty a_m(k+g) e^{\I m \Omega x} \label{eq:four3}\\
&= \sum_{m=-\infty}^{\infty} a_m(k) e^{\I (m-n) \Omega x} \\
&= \sum_{l=-\infty}^\infty a_{l+n}(k) e^{\I l \Omega x} \label{eq:four4}.
\end{align}
It follows from \eqref{eq:four3} and \eqref{eq:four4} that $a_m$ satisfies 
the functional relation
\begin{equation}
a_m(k+n \Omega) = a_{m+n}(k), \quad m, n \in \mathbb{Z},\, k \in \mathbb{R}.
\label{eq:funcrel}
\end{equation}
Setting
\begin{equation}
\alpha(k+m \Omega) = a_m(k), \quad k \in \Ibz,\, m \in \mathbb{Z} 
\end{equation}
yields the result.
\end{proof}
Lemma \ref{lemma:afun} shows that we can obtain the function $\alpha$ defined on $\R$ by ``unfolding'' the Fourier coefficients $a_m(k)$ defined on $\Ibz$ for $m\in\mathbb{Z}$ such that $a_m(k)$ takes up the interval $[-\Omega/2 +m\Omega, \Omega/2 +m\Omega)$. We summarize this fact in the following lemma, together with Lemma \ref{lemma:aman} and the decay condition (\ref{eq:decay}) in terms of $\alpha$.

\begin{lemma}
Suppose that $u(\cdot,k) \in C^1(\Iuc;\C)$ is periodic and $\partial_x u(\cdot,k)$ is of bounded variation for $k\in\Ibz$. The Fourier series \eqref{eq:ufour} can
be expressed as
\begin{equation}
u(x,k) = \sum_{m=-\infty}^\infty \alpha(k+m \Omega) e^{\I m \Omega x},
\label{eq:ufour2}
\end{equation}
where $\alpha$ is defined in (\ref{eq:afun}). The function $\alpha(k)$ can be chosen to be analytic on the intervals $[-\Omega/2 + m\Omega,\Omega/2+m\Omega)$ for $m\in\mathbb{Z}$. Furthermore, $\abs{\alpha(k)}$ decays no slower than $|k|^{-2}$ as $|k| \to \infty$.

\label{lemma:ufour}
\end{lemma}
\begin{remark}
Since $u(x,k)$ may not satisfy (\ref{eq:perk3}) (see Remark \ref{rmk:nonper}),
the function $\alpha$ can be a discontinuous function in $\R$. Lemma \ref{lemma:aman} only ensures $\alpha$ is analytic in each interval $[-\Omega/2 +m\Omega, \Omega/2 +m\Omega)$ for $m\in \mathbb{Z}$. 
Moreover, the function $\alpha$ is formed from a very special choice of coefficients $a_m$ since the corresponding $u(x,k)$ is \emph{chosen} to be analytic in $k$; this does not happen if one computes $u(x,k)$ at each point $k$ independently (see Remark \ref{rmk:phase}) so that $u(x,k)$ is an arbitrarily irregular function in $k$. 
\label{rmk:bad}
\end{remark}

\subsection{Properties of Wannier functions}
In this section, we introduce the Fourier transform of the Wannier functions introduced in Section \ref{sec:wannier}, and their variance in terms of the $\alpha$ function defined in (\ref{eq:ufour2}).
Following the discussion above, we suppress the band index and we rewrite (\ref{eq:wannier}) as
\begin{equation}
W_0(x) = \frac{1}{\Omega}\int_{-\Omega/2}^{\Omega/2} \psi_k(x)\,{\rm d}k = \frac{1}{\Omega}\int_{-\Omega/2}^{\Omega/2} e^{\I xk}u_k(x)\,{\rm d}k.
\label{eq:wannier2}
\end{equation}
First, we introduce the Fourier transform of $W_0(x)$. 
By substituting (\ref{eq:ufour2}) into (\ref{eq:wannier2}), we obtain
\begin{equation}
	W_0(x) = 
\frac{1}{\Omega}\int_{-\Omega/2}^{\Omega/2}
\sum_{m=-\infty}^\infty \alpha(k+m \Omega) e^{\I (k + m\Omega ) x} \D k
= \frac{1}{2\pi} \int_{-\infty}^{\infty} \frac{2\pi}{\Omega} \alpha(\xi)e^{\I
\xi x} \D \xi.
	\label{eq:wanift}
\end{equation}
We observe that (\ref{eq:wanift}) only holds formally since $\alpha$ may not even be integrable (see Remark \ref{rmk:bad}). However, whenever $\alpha$ is chosen to be integrable, (\ref{eq:wanift}) shows that $W_0$ in (\ref{eq:wannier2}) is the inverse Fourier transform of $\frac{2\pi}{\Omega}\alpha$ . Furthermore, if we assume that $\alpha$ is chosen such that both $\alpha'$ and $\alpha''$ are integrable, $\abs{W_0(x)}$  will decay no slower than $1/\abs{x}^2$ so that the Fourier inversion formula is applicable. Thus we have the following observation.
\begin{theorem}
Suppose $u$ in Lemma \ref{lemma:ufour} is chosen such that $\alpha$ is integrable.
Then the corresponding Wannier function $W_0$ in (\ref{eq:wannier2}) is well defined and is given by the formula
\begin{equation}
	W_0(x) =\frac{1}{2\pi} \int_{-\infty}^{\infty} \frac{2\pi}{\Omega}\,\alpha(\xi)e^{\I \xi x} \D \xi.
	\label{eq:wanift2}
\end{equation}
By the Riemann--Lebesgue lemma, $\abs{W_0(x)}$ goes to zero as $\abs{x} \to \infty$.
Suppose further that $\alpha'$ and $\alpha''$ are also integrable. Then the Fourier transform of $W_0$, denoted by $\widehat{W}_0$, is given by the formula
\begin{equation}
\widehat{W}_0(\xi) = \int_{-\infty}^\infty W_0(x)e^{-\I \xi x}\,{\rm d}x = \frac{2\pi}{\Omega}  \alpha(\xi).
\label{eq:wannfour}
\end{equation}

\label{thm:wannfour}
\end{theorem}
Thus, if $\alpha$ is chosen to have poor regularity, the Wannier function $W_0$
is
poorly
localized. To define its variance $\alpha$ needs to be smooth so that
$\abs{W_0(x)}$ goes to zero faster than $1/\abs{x}$ as $\abs{x} \to \infty$. 
As a result,  we also assume that both $\alpha'$ and $\alpha''$ are integrable.
The following lemma contains the formulas for the first and second moments of
$W_0$ in terms of $\alpha$ for defining its variance. They are consequences of
the one-dimensional case of (7) and (8) in \cite{marzari1997maximally}, first
derived in \cite{blount1962formalisms}. 

\begin{lemma}
Suppose that $u$ in Lemma \ref{lemma:ufour} is chosen such that the function $\alpha$ has integrable derivatives $\alpha'$ and $\alpha''$.
We have the following formulas for the first and second moments of $W_0$:
\begin{equation}
\langle x \rangle = \int_{-\infty}^\infty x \overline{W}_0(x) W_0(x)\,{\rm d}x = \frac{2\pi}{\Omega^2}\,\I \int_{-\infty}^\infty \overline{\alpha}(\xi)\alpha'(\xi) \,{\rm d}\xi,
\label{eq:meanid2}
\end{equation}
and
\begin{equation}
\langle x^2 \rangle = \int_{-\infty}^\infty x^2 \overline{W}_0(x) W_0(x)\,{\rm d}x = \frac{2\pi}{\Omega^2} \int_{-\infty}^\infty |\alpha'(\xi)|^2 {\rm d} \xi.
\label{eq:varid2}
\end{equation}
Hence the variance $\langle x^2 \rangle - \langle x \rangle^2$ is given by the formula
\begin{equation}
	\langle x^2 \rangle - \langle x \rangle^2 =  \frac{2\pi}{\Omega^2} \int_{-\infty}^\infty |\alpha'(\xi)|^2 {\rm d} \xi + \frac{4\pi^2}{\Omega^4} \bb{\int_{-\infty}^\infty \overline{\alpha}(\xi)\alpha'(\xi) \,{\rm d}\xi}^2.
	\label{eq:varwan}
\end{equation}
\label{lemma:momentid}
\end{lemma}

\subsection{Perturbation analysis in the Fourier domain}
\label{s:odewannier}
The function $u$ in Lemma \ref{lemma:ufour} can be obtained by applying Kato's construction in Theorem \ref{thm:katovec}. This involves the ODE defined by (\ref{eq:evalpert}) for the eigenvalue problem defined by (\ref{eq:udiff}). In order to carry out this procedure computationally, we transform (\ref{eq:udiff}) into the Fourier domain. In this section, we introduced the Fourier domain version of (\ref{eq:udiff}) and apply (\ref{eq:evalpert}) to the result.

First, we introduce Fourier series of the potential $V$ given by the formula
\begin{equation}
V(x) = \sum_{l=-\infty}^\infty \widehat{V}_l e^{\I l \Omega x},
\label{eq:vfour}
\end{equation}
where the vector with elements $\{\widehat{V}_l \}_{l=-\infty}^\infty$ is in $\ell^2(\mathbb{Z})$ for piecewise continuous $V$ in $\Iuc$.
Moreover, we define a vector $\mathbf{y}$ containing the Fourier coefficients of $u$ in (\ref{eq:ufour}) and (\ref{eq:ufour2}):
\begin{equation}
	\mathbf{y}_m(k) = a_m(k) = \alpha(k+m \Omega), \quad m \in \mathbb{Z},\, k\in\Ibz.
	\label{eq:yvec}
\end{equation}
It follows that $\mathbf{y}\in \ell^2(\mathbb{Z})$ for any $u$ that satisfies the
assumptions of Lemma \ref{lemma:ufour}. 
Moreover, it follows from \eqref{eq:anorm} that
\begin{equation}
	\norm{\mathbf{y}(k)}^2 = \frac{1}{a}, \quad k\in\Ibz.
	\label{eq:ynorm}
\end{equation}

Inserting \eqref{eq:ufour} into \eqref{eq:udiff} yields the infinite linear
system 
\begin{equation}
(\mathbf{D}(k) + \mathbf{V}) \mathbf{y}(k) = E(k)\mathbf{y}(k),
\quad k \in \Ibz,
\label{eq:linsys}
\end{equation}
where for all $m, n \in \mathbb{Z}$ we have
\begin{align}
\mathbf{D}_{m,n}(k) &= (k+m \Omega)^2 \delta_{mn},
\label{eq:dmat}\\
\mathbf{V}_{m,n} &= \widehat{V}_{m-n}.
\label{eq:vmat}
\end{align}
It is convenient to define 
\begin{equation}
	\mathbf{\Theta}(k) = \mathbf{D}(k) + \mathbf{V} - E(k)\mathbf{I},	
	\label{eq:theta}
\end{equation}
where $\mathbf{I}_{m,n}=\delta_{mn}$ is the identity.
Applying (\ref{eq:evalpert}) and (\ref{eq:evecpert}) to (\ref{eq:linsys}), we obtain the following formulas
\begin{equation}
	\frac{\D}{\D k} E(k) = \frac{\mathbf{y}(k)^* \mathbf{S}(k) \mathbf{y}(k)}{\mathbf{y}(k)^*\mathbf{y}(k)},
	\label{eq:evalf}
\end{equation}
\begin{equation}
	\frac{\D}{\D k} \mathbf{y}(k) = -\mathbf{\Theta}^\dagger(k) \mathbf{S}(k) \mathbf{y}(k),\quad k\in\Ibz,
	\label{eq:evecf}
\end{equation}
where
\begin{align}
\mathbf{S}_{m,n}(k) &= 2(k+m \Omega) \delta_{mn}.
\label{eq:smat}
\end{align}
The significance of (\ref{eq:evalf}) and (\ref{eq:evecf}) is that they provide
an infinite system of ODEs for the energy $E$ and the vector containing the
Fourier coefficients of $u$. 
If $E_0$ and $\mathbf{y}_0$ satisfy \eqref{eq:linsys} at $k_0 = -\Omega/2$,
solving \eqref{eq:evalf} and \eqref{eq:evecf}for $k \in \Ibz$ subject to the
initial conditions
\begin{equation}
	E(k_0) = E_0, \quad \mathbf{y}(k_0) = \mathbf{y}_0
	\label{eq:incon}
\end{equation}
where $k_0 = -\Omega/2$, produces $E$ and $u$ that are analytic functions of $k$.

It is worth mentioning that (\ref{eq:evecf}) is sometimes referred to as the
parallel-transport equation due to the relation
\begin{equation}
	\mathbf{y}(k)^* \frac{\D}{\D k} \mathbf{y}(k) =0,
	\label{eq:parallel}
\end{equation}
which holds due to the relation 
\begin{equation}
	\mathbf{\Theta}^\dagger(k)\widetilde{\mathbf{y}}(k) = 0,
	\label{eq:ortho}
\end{equation}
which in turn is a consequence of \eqref{eq:orthopv}.
In other words, $\frac{\D}{\D k} \mathbf{y}(k)$ is always orthogonal to
$\mathbf{y}(k)$. 
The relation (\ref{eq:parallel}) also implies that $\norm{\mathbf{y}(k)}$ is
constant for all $k$.

We observe that the operator $-\mathbf{\Theta}^\dagger(k) \mathbf{S}(k)$ in (\ref{eq:evecf}) has a spectrum that decays as $1/m$ asymptotically for any $k\in\Ibz$, so it is a compact operator on $\ell^2(\mathbb{Z})$. As a result, despite the fact that (\ref{eq:evecf}) is infinite-dimensional, (\ref{eq:evecf}) is very benign both mathematically and numerically. 

\subsection{Analysis of the discontinuities at zone boundaries\label{sec:dis}}
In this section, we show that the discontinuities of the Fourier transform $\alpha(k)$ (and its derivatives of all orders) at $k= \Omega/2 + m\Omega$ for $m\in\mathbb{Z}$ are identical.

Suppose that $\mathbf{y}(k)$ for $k\in\Ibz$ is a solution to the ODE (\ref{eq:evecf}) with initial condition $(\ref{eq:incon})$ at $k_0 = -\Omega/2$. The vector $\mathbf{y}(k)$ defines a function $\alpha(k)$ (see (\ref{eq:yvec})) by
\begin{equation}
	\mathbf{y}_m(k) = a_m(k) = \alpha(k+m \Omega), \quad m \in \mathbb{Z},\, k\in\Ibz.
	\label{eq:yvec2}
\end{equation}
We can extend $\mathbf{y}(k)$ to $k=\Omega/2$, so that we have
\begin{equation}
	\mathbf{y}_m(\Omega/2)  = a_m(\Omega/2), 	\quad m \in \mathbb{Z}.
\end{equation}
As discussed in Remark \ref{rmk:bad}, $a_m(\Omega/2)$ is not guaranteed to be equal to $a_{m+1}(-\Omega/2)$. 
Since $a_m(\Omega/2)$ and $a_{m+1}(-\Omega/2)$ may be viewed as limits of
$\alpha(k)$ at the boundary points $k=\Omega/2+m\Omega$ for $m\in\mathbb{Z}$ from
both sides, 
their differences characterize the jump of $\alpha(k)$ at these boundary points. 
The following lemma shows that $a_m(\Omega/2)$ and $a_{m+1}(-\Omega/2)$
only differ by a phase factor independent of $m$.

\begin{lemma}
	For the vector $\mathbf{y}(k)$ with $k\in[-\Omega/2,\Omega/2]$ defined above, we have the following relation
	\begin{equation}
		a_{m}(\Omega/2) = e^{\I \varphi_{\rm zak}} a_{m+1}(-\Omega/2), \quad m \in \mathbb{Z},
		\label{eq:diff}
	\end{equation}
where $\varphi_{\rm zak}$ is a real number independent of $m$.
\label{lem:zak}
\end{lemma}
\begin{proof}
	We observe that both vectors $\mathbf{y}(-\Omega/2)$ and $\mathbf{y}(\Omega/2)$ satisfy (\ref{eq:linsys}):
\begin{equation}
(\mathbf{D}(-\Omega/2) + \mathbf{V}) \mathbf{y}(-\Omega/2) = E\mathbf{y}(-\Omega/2),
\label{eq:linsys2}
\end{equation}
\begin{equation}
(\mathbf{D}(\Omega/2) + \mathbf{V}) \mathbf{y}(\Omega/2) = E\mathbf{y}(\Omega/2),
\label{eq:linsys3}
\end{equation}
where $E(-\Omega/2) = E(\Omega/2)=E$ (see Remark \ref{rmk:nonper}). 
Exploiting the structure in (\ref{eq:dmat}) and (\ref{eq:vmat}), we can reindex
(\ref{eq:linsys3}) to yield 
\begin{equation}
(\mathbf{D}(-\Omega/2) + \mathbf{V}) R(\mathbf{y}(\Omega/2)) = E R(\mathbf{y}(\Omega/2)),
\label{eq:linsys4}
\end{equation}
where $R$ is the right shift operator defined in (\ref{eq:rshift}). Since the eigenvalue $E$ is not degenerate, $\mathbf{y}(-\Omega/2)$ and $R(\mathbf{y}(\Omega/2))$ can only differ by a constant of unity modulus $e^{\I \varphi_{\rm zak}}$ for some real $\varphi_{\rm zak}$. The component form of $e^{\I \varphi_{\rm zak}}\mathbf{y}(-\Omega/2) =  R(\mathbf{y}(\Omega/2))$ yields (\ref{eq:diff}).
\end{proof}
A similar argument applied to (\ref{eq:evecf}) shows that 
\begin{equation}
	\frac{\D}{\D k} \mathbf{y}(-\Omega/2) = -\mathbf{\Theta}^\dagger(-\Omega/2) \mathbf{S}(-\Omega/2) \mathbf{y}(-\Omega/2),
	\label{eq:evecf2}
\end{equation}
\begin{align}
	\frac{\D}{\D k} R(\mathbf{y}(\Omega/2) )&= -\mathbf{\Theta}^\dagger(-\Omega/2) \mathbf{S}(-\Omega/2) R( \mathbf{y}(\Omega/2)) 
	\label{eq:evecf3} \\
&= 
-e^{\I \varphi_{\rm zak}} \mathbf{\Theta}^\dagger(-\Omega/2) \mathbf{S}(-\Omega/2)\mathbf{y}(-\Omega/2) \\
&= e^{\I \varphi_{\rm zak}}\frac{\D}{\D k} \mathbf{y}(-\Omega/2).
\end{align}
This shows that the jump between $a'_m(\Omega/2)$ and $a'_{m+1}(-\Omega/2)$ is also a factor of $e^{\I \varphi_{\rm zak}}$. 
Repeated differentiation of (\ref{eq:evecf2}) and (\ref{eq:evecf3}) shows the jump between the $n$th derivatives $a^{(n)}_m(\Omega/2)$ and $a^{(n)}_{m+1}(-\Omega/2)$ are all identical. 
The quantity $\varphi_{\rm zak}$ was introduced by Zak in \cite{zak1989berry}
in a different context and is known as the Zak phase.
It was shown in \cite{zak1989berry} that the Zak phase depends only on the Schr\"odinger equation \eqref{eq:schro} and determines the center of the Wannier function. 

\subsection{Gauge transformation}
In this section, we introduce gauge transformation and show that a simple class of gauge choices fixes all discontinuities of the Fourier transform $\alpha$ discussed in Section\,\ref{sec:dis}, thus obtaining exponentially localized Wannier functions.

Given $\varphi \in C^1(\Ibz; \R)$ (not necessarily periodic), we can modify any
solution $\mathbf{y}$ of $(\ref{eq:evecf})$ via 
\begin{equation}
\widetilde{\mathbf{y}}(k) = e^{-\I \varphi(k)} \mathbf{y}(k) \,,\quad \mbox{ for $k\in\Ibz$ },
\label{eq:yvectilde}
\end{equation}
where $\widetilde{\mathbf{y}}(k)$ still satisfies (\ref{eq:linsys}) with the same $E(k)$ as $\mathbf{y}(k)$ and the normalization (\ref{eq:ynorm}).
Such a transform turns $\alpha$ defined by $\mathbf{y}$ into
$\widetilde{\alpha}$ given by the formula
\begin{equation}
	\widetilde{\alpha}(k+m \Omega)= e^{-\I \varphi(k)}\mathbf{y}_m(k) = e^{-\I \varphi(k)}a_m(k) , \quad m \in \mathbb{Z},\, k\in\Ibz.
	\label{eq:afuntilde}
\end{equation}
This process is known as gauge transformation in the physics literature. 
From \eqref{eq:evecf}, it follows that $\widetilde{\mathbf{y}}$ satisfies
\begin{equation}
\widetilde{\mathbf{y}}'(k)  = -\mathbf{\Theta}^\dagger(k) \mathbf{S}(k)
\widetilde{\mathbf{y}}(k) - \I \varphi'(k) \widetilde{\mathbf{y}}(k), \quad k\in\Ibz.
\label{eq:tildeode}
\end{equation}
In the physics literature, $\varphi'$ is known as the Berry connection and the phase $\varphi$ is known as the Berry phase. 
The Berry connection is usually denoted by $A$ and defined as $A(k)=\I\braket{u_k}{u'_k}$ in Bra-Ket notation for some function $u_k$ satisfying (\ref{eq:ueq}). 
In this paper, this corresponds to the identity: 
\begin{equation}
	\I\widetilde{\mathbf{y}}^*(k)\widetilde{\mathbf{y}}'(k) = \varphi'(k)/a,
	\label{eq:berrycon}
\end{equation}
where we used (\ref{eq:ortho}) and the factor $1/a$ comes from the normalization (\ref{eq:ynorm}).

We choose $\varphi$ such that $\varphi(k) = -\varphi(-k)$ and $\varphi'(k) = \varphi_{\rm zak}/\Omega$, where $\varphi_{\rm zak}$ is defined in (\ref{eq:diff}):
\begin{equation}
	\varphi(k) = \int_{-\Omega/2}^{k} \varphi'(k') \D k' - \frac{\Omega}{2} = \frac{\varphi_{\rm zak}}{\Omega} k, \quad k\in\Ibz.
	\label{eq:minzak}
\end{equation}
By Lemma \ref{lem:zak}, this choice of $\varphi(k)$ fixes the jump so that 
\begin{equation}
	 \widetilde{\mathbf{y}}(-\Omega/2) = R(\widetilde{\mathbf{y}}(\Omega/2)).
	 \label{eq:tildejump}
\end{equation}
Furthermore, analogously to (\ref{eq:evecf2}) and (\ref{eq:evecf3}), we have 
\begin{equation}
\widetilde{\mathbf{y}}'(-\Omega/2)  = -\mathbf{\Theta}^\dagger(-\Omega/2) \mathbf{S}(-\Omega/2)
\widetilde{\mathbf{y}}(-\Omega/2) - \I \varphi'(-\Omega/2) \widetilde{\mathbf{y}}(-\Omega/2), 
\label{eq:tildeode2}
\end{equation}
\begin{equation}
R(\widetilde{\mathbf{y}}'(\Omega/2))  = -\mathbf{\Theta}^\dagger(-\Omega/2) \mathbf{S}(-\Omega/2)
R(\widetilde{\mathbf{y}}(\Omega/2)) - \I \varphi'(\Omega/2) R(\widetilde{\mathbf{y}}(\Omega/2)). 
\label{eq:tildeode3}
\end{equation}
Since $\varphi'$ in (\ref{eq:minzak}) is a constant, we have $\varphi'(-\Omega/2) = \varphi'(\Omega/2)$. Together with (\ref{eq:tildejump}), we conclude that 
\begin{equation}
	\widetilde{\mathbf{y}}'(-\Omega/2) = R(\widetilde{\mathbf{y}}'(\Omega/2)).
\end{equation}
By repeated differentiation of (\ref{eq:tildeode2}) and (\ref{eq:tildeode3}), we conclude that 
\begin{equation}
	\widetilde{\mathbf{y}}^{(n)}(-\Omega/2) = R(\widetilde{\mathbf{y}}^{(n)}(\Omega/2)),
	\label{eq:tildejumpn}
\end{equation}
holds for all derivatives of order $n\ge 1$. From (\ref{eq:tildejump}), (\ref{eq:tildejumpn}) and Lemma \ref{lemma:ufour}, we conclude that the transformed function $\widetilde{\alpha}$ in (\ref{eq:afuntilde}) corresponding to $\widetilde{\mathbf{y}}(k)$ is analytic on $\R$. Thus, its corresponding Wannier function defined by (\ref{eq:wanift2}) will be exponentially localized.

We also observe that the choice $\varphi'(k) = \varphi_{\rm zak}/\Omega$ is not the only possibility that yields exponential localization. 
The above will still hold if $\varphi'$ is chosen to be a real, periodically analytic function on $\Ibz$ with the zeroth Fourier coefficient $(\varphi_{\rm zak}+2\pi n)/\Omega$ for any $n\in\mathbb{Z}$. 
Such functions can be written as 
\begin{equation}
	\varphi'(k) = \frac{\varphi_{\rm zak}+2\pi n}{\Omega} + \sum_{\substack{m=-\infty \\ m \ne 0}}^{\infty} c_m e^{\I m a k}, 
	\label{eq:zakgen}
\end{equation} 
where the Fourier coefficients $c_m$ decay exponentially as $\abs{m} \to
\infty$, and the Berry phase $\varphi$ is then given by 
\begin{equation}
	\varphi(k) = \int_{-\Omega/2}^{k} \varphi'(s)\,\D s - \frac{\Omega}{2},\quad k\in\Ibz.
	\label{eq:zakgenp}
\end{equation} 
We summarize this fact in the following theorem. It will be the key tool for constructing exponentially localized Wannier functions.
\begin{theorem}
	Suppose that $\mathbf{y}$  on $\Ibz$ is a solution to (\ref{eq:evecf}) with initial conditions (\ref{eq:incon}) and
$\mathbf{y}$ defines a function $\alpha$ on $\R$ by (\ref{eq:yvec}).
Suppose further that $\varphi'$ is chosen as (\ref{eq:zakgen}), by which we apply the gauge transform (\ref{eq:yvectilde}) to define a new function $\widetilde{\mathbf{y}}$. Then the function $\widetilde{\alpha}$ defined by $\widetilde{\mathbf{y}}$ in (\ref{eq:afuntilde}) is analytic on $\R$. As a result, the Wannier function as the Fourier transform of $\widetilde{\alpha}$ in (\ref{eq:wanift2}) is exponentially localized.
\label{thm:wanexp}
\end{theorem}

\subsection{Gauge choice and Wannier localization\label{sec:opt}}
In the previous section, we have seen that many choices of the Berry phase $\varphi$ (see Theorem \ref{thm:wanexp} and (\ref{eq:zakgen})) are able to produce exponentially localized Wannier functions. In this section, we derive the optimal choice in terms of minimizing the variance defined by (\ref{eq:varwan}). It is a standard measure of the localization of Wannier functions \cite{marzari1997maximally}. First, we modify the formulas in Lemma \ref{lemma:momentid} to account for the inclusion of $\varphi$. The following is a result of (\ref{eq:afuntilde}), (\ref{eq:tildeode}), (\ref{eq:berrycon}) with $\varphi'$ given in (\ref{eq:zakgen}) and the formulas in Lemma \ref{lemma:momentid}. 

\begin{lemma}
Suppose that $\widetilde{\mathbf{y}}$ is defined in Theorem \ref{thm:wanexp}. We have the following formulas
\begin{equation}
\langle x \rangle = \frac{a}{2\pi}\sb{\varphi \left( \frac{\Omega}{2} \right) - \varphi \left( -\frac{\Omega}{2} \right)} = \frac{\varphi_{\rm zak}}{2\pi}a + na,
\label{eq:tildemoment}
\end{equation}
\begin{eqnarray}
\langle x^2 \rangle &=& \frac{a^2}{2\pi}\int_{-\Omega/2}^{\Omega/2} \norm{ \mathbf{y}'(k) }^2 \,{\rm d}k + \frac{a}{2\pi}\int_{-\Omega/2}^{\Omega/2} \varphi'(k)^2\,{\rm d}k \\
&=& \frac{a^2}{2\pi}\int_{-\Omega/2}^{\Omega/2} \norm{ \mathbf{y}'(k) }^2 \,{\rm d}k + \bb{\frac{\varphi_{\rm zak}}{2\pi}a + na}^2 + \sum_{\substack{m=-\infty \\ m \ne 0}}^{\infty} \abs{c_m}^2,
\label{eq:tildemoment2}
\end{eqnarray}
and
\begin{equation}
	\langle x^2 \rangle  - \langle x \rangle ^2 = \frac{a^2}{2\pi}\int_{-\Omega/2}^{\Omega/2} \norm{ \mathbf{y}'(k) }^2 \,{\rm d}k + \sum_{\substack{m=-\infty \\ m \ne 0}}^{\infty} \abs{c_m}^2.
	\label{eq:varopt}
\end{equation}
\label{lemma:varmin}
\end{lemma}
As a consequence, we see that the optimally localized Wannier function for a single band is constructed based on the gauge choice
\begin{equation}
	\varphi'(k) = \frac{\varphi_{\rm zak}+2\pi n}{\Omega}, 
	\label{eq:zakopt}
\end{equation} 
i.e., the choice of (\ref{eq:zakgen}) with all of the coefficients $c_m$ set to zero. 
For such a choice, we observe that the variance in (\ref{eq:varopt}) is gauge-independent since $\mathbf{y}$ is constructed by solving (\ref{eq:evecf}), which only depends on the original Schr\"odinger equation (\ref{eq:schro}). 

We observe that the integer $n$ in (\ref{eq:zakopt}) only results in a shift of the center (\ref{eq:tildemoment}). Since every lattice point has an identical copy of the Wannier function (see (\ref{eq:wancopy})), any nonzero choice of $n$ amounts to shifting the index of the Wannier functions. 
As a result, we simply set $n=0$, which is the choice in
(\ref{eq:minzak}).  
We thus have the following corollary.
\begin{corollary}
	An optimal choice of gauge in term of the variance of Wannier
functions is given by 
\begin{equation}
	\varphi(k) = \int_{-\Omega/2}^{k} \varphi'(s)\, \D s - \frac{\Omega}{2} = \frac{\varphi_{\rm zak}}{\Omega} k,\quad k\in\Ibz.
	\label{eq:minzak2}
\end{equation}
The center $\langle x \rangle$ and the variance $\langle x^2 \rangle  - \langle x \rangle ^2$ of the optimally localized Wannier function are given by the formulas
\begin{equation}
	\langle x \rangle = \frac{\varphi_{\rm zak}}{2\pi}a.
\end{equation}
\begin{equation}
\langle x^2 \rangle  - \langle x \rangle ^2 = \frac{a^2}{2\pi}\int_{-\Omega/2}^{\Omega/2} \norm{ \mathbf{y}'(k) }^2 \,{\rm d}k.
\end{equation} 
Both quantities only depend on the given Schr\"odinger equation.
\label{cor:opt}
\end{corollary}
\noindent Theorem \ref{thm:wanexp} and Corollary \ref{cor:opt} complete the construction of optimally localized Wannier functions.

We note that the condition in (\ref{eq:zakopt}) can also be obtained by applying the calculus of variations on the variance in (\ref{eq:varopt}).
Since the center in (\ref{eq:varopt}) must equal  $\frac{\varphi_{\rm zak}}{2\pi}a + na$ for (\ref{eq:tildejump}) to hold, it can be excluded from the objective function.
Therefore, to minimize the variance it is sufficient for $\varphi$ to minimize $\langle x^2 \rangle$, where the only gauge-dependent term is
\begin{equation}
\int_{-\Omega/2}^{\Omega/2} \varphi'(k)^2\,{\rm d}k.
\label{eq:laplace}
\end{equation}
The optimal solution is characterized by the solution to the Laplace equation
\begin{equation}
	\varphi'' = 0,\quad \mbox{subject to $\varphi(\Omega/2) - \varphi(-\Omega/2) = \frac{\varphi_{\rm zak}}{2\pi}a + na$.}
\end{equation}
The derivative of the solution is precisely (\ref{eq:zakopt}). 
Results of
this form are known in physics literature \cite{blount1962formalisms,marzari1997maximally}. 

\subsection{Realty of Wannier functions}
While we have completed the construction of optimally localized Wannier
functions, there is one remaining degree of freedom,
the phase choice of the vector $\mathbf{y}_0$ in (\ref{eq:incon}); 
we could replace $\mathbf{y}_0$ by $e^{-\I \varphi_0}\mathbf{y}_0$ 
for any real $\varphi_0$ and all above formulas would still hold, including
the expressions in Corollary \ref{cor:opt}. 
However, the phase factor $e^{-\I \varphi_0}$ does change the resulting $W_0$. 
In this section, we prove that the constructed Wannier function $W_0$ can
always be chosen to be real if $\mathbf{y}_0$ in (\ref{eq:incon})
satisfies a simple condition.

First, we observe that, since the potential $V$ is real, we have 
\begin{equation}
	\widehat{V}_m = \overline{\widehat{V}}_{-m}, \quad m\in\mathbb{Z}.
	\label{eq:vsym}
\end{equation}
Moreover, the symmetry in (\ref{eq:trsym}) implies that the energy $E$ satisfies 
\begin{equation}
	E(k) = E(-k), \quad k\in\Ibz.
	\label{eq:esym}
\end{equation}
In the following, we show that these two symmetries allow us to run the construction backward in $k$ to show the realty of $W_0$; 
it is a consequence of the uniqueness of infinite-dimensional, linear
initial value problems with compact coefficients (cf.\ e.g., \cite{deimling2006ordinary}). 

\begin{theorem}
\label{thm:real}
	Suppose the vector $\widetilde{\mathbf{y}}(k)$ for $k\in\Ibz$ and its corresponding function $\widetilde{\alpha}$ are constructed based on Theorem \ref{thm:wanexp} with $\varphi$ given in Corollary \ref{cor:opt}. Suppose further that the vector $\mathbf{y}(k)=\cb{a_m(k)}_{m=-\infty}^{\infty}$ is related to $\widetilde{\mathbf{y}}(k)$ by (\ref{eq:yvectilde}). Then it is always possible to choose the elements in the initial vector $\mathbf{y}(-\Omega/2)=\mathbf{y}_0$ in (\ref{eq:incon}) to satisfy
	\begin{equation}
			a_m(-\Omega/2) = \overline{a}_{-m}(\Omega/2), \quad m\in\mathbb{Z},
		\label{eq:realcon}
	\end{equation}
so that the function $\widetilde{\alpha}$ satisfies
	\begin{equation}
		\widetilde{\alpha}(k+m \Omega) 
		=\overline{\widetilde{\alpha}}(-k-m \Omega),\quad k\in\Ibz,\,m\in\mathbb{Z}.
	\end{equation}
As a result, the Wannier function $W_0$ corresponding to $\widetilde{\alpha}$ can always be chosen to be real.
\end{theorem}
\begin{proof}
Suppose that $\mathbf{y}(k)=\cb{a_{m}(k)}_{m=-\infty}^{\infty}$  satisfies
the ODE (\ref{eq:evecf}).
We first show that there is a choice of $\mathbf{y}_0$ such that if
$\mathbf{y}$ also satisfies \eqref{eq:incon}, then
\begin{equation}
	a_m(-\Omega/2) = \overline{a}_{-m}(\Omega/2), \quad m\in\mathbb{Z}.
 	\label{eq:realcon2}
\end{equation}
For now, we let $\mathbf{y}_0$ be an arbitary solution to \eqref{eq:linsys}
with $k=-\Omega/2$.

By (\ref{eq:dmat}), (\ref{eq:vmat}) and (\ref{eq:vsym}), we
observe that the operator $\mathbf{D}(k)+\mathbf{V}$ in (\ref{eq:linsys}) is
invariant under the operation consisting of the substitutions $m\rightarrow
-m$, $n\rightarrow -n$, $k\rightarrow -k$ and complex conjugation. 
Applying such a transformation to \eqref{eq:linsys}, together with
(\ref{eq:esym}), gives
\begin{equation}
	(\mathbf{D}(k)+\mathbf{V})\mathbf{x}(k) = E(k)\mathbf{x}(k),
\end{equation}
where $\mathbf{x}(k)=\cb{\overline{a}_{-m}(-k)}_{m=-\infty}^{\infty}$. 
Since $E(k)$ is non-degenerate and $\| \mathbf{y}(k) \|$ is constant, there
must exist a real $\theta$ such that
\begin{equation}
 a_m(-\Omega/2) = e^{\I \theta} \overline{a}_{-m}(\Omega/2), \quad m \in \mathbb{Z}.
\end{equation}
Replacing the original $\mathbf{y}_0$ with $e^{-\I \varphi_0}  \mathbf{y}_0$ with $\varphi_0=\theta/2$ produces $\mathbf{y}$ that satisfies \eqref{eq:realcon2}. 


Next, suppose that $\mathbf{y}(k)=\cb{a_{m}(k)}_{m=-\infty}^{\infty}$ is a
solution to the ODE \eqref{eq:evecf} satisfying \eqref{eq:realcon2}.
Applying the aforementioned transformation to \eqref{eq:evecf} yields 
\begin{equation}
			\frac{\D}{\D k} \mathbf{x}(k) = -\mathbf{\Theta}^\dagger(k) \mathbf{S}(k) \mathbf{x}(k)\,,
\end{equation}
where again we have $\mathbf{x}(k)=\cb{\overline{a}_{-m}(-k)}_{m=-\infty}^{\infty}$. 
Noting that \eqref{eq:realcon2} is equivalent to the condition that $\mathbf{y}(-\Omega/2) = \mathbf{x}(-\Omega/2)$, it follows from the uniqueness theorem for initial value problems that $\mathbf{x}(k) = \mathbf{y}(k)$ for $k\in [-\Omega/2,\Omega/2]$. 
Consequently, we have
\begin{equation}
	a_m (k) = \overline{a}_{-m}(-k),\quad k\in[-\Omega/2,\Omega/2],\,m\in\mathbb{Z}.
\end{equation}
Together with (\ref{eq:afuntilde}) and $\varphi(k)=-\varphi(-k)$ in (\ref{eq:minzak2}), the above relation yields the desired result:
\begin{equation}
	\widetilde{\alpha}(k+m \Omega)=  e^{-\I \varphi(k)}a_m(k) = e^{\I \varphi(-k)}\overline{a}_{-m}(-k) = \overline{\widetilde{\alpha}}(-k-m \Omega)\quad m\in\mathbb{Z}.
\end{equation}

\end{proof}
We observe that the condition in (\ref{eq:realcon}) is applicable after solving (\ref{eq:evecf}) for determining $\mathbf{y}$ on $[-\Omega/2,\Omega/2]$; we can obtain a constant phase factor $e^{-\I \varphi_0}$ to enforce (\ref{eq:realcon}) so the the resulting $W_0$ is real.

\section{Constructing optimally localized Wannier functions \label{s:procedure}}
In this section, we use the apparatus in Section \ref{s:apparatus} to construct optimally localized Wannier function.
Suppose we are given a real, piecewise continuous potential $V$ with period given by the lattice constant $a$. 
Let $\Omega$ be defined by \eqref{eq:om22} and $\Ibz$ by \eqref{eq:ibz22}. 
We also select the band for which the Wannier function will be
constructed. 

The construction can be divided into three stages. 
First, we compute the Fourier series of $V$ as in (\ref{eq:vfour}) and transform (\ref{eq:udiff}) into its Fourier version in (\ref{eq:linsys}). Then, we compute the energy $E_0$ and the eigenvector $\mathbf{y}_0$ at $k=-\Omega/2$ for the selected band. 
Second, we solve the ODE defined by (\ref{eq:evalf}) and  (\ref{eq:evecf}) in $[-\Omega/2,\Omega/2]$ with initial conditions given by $E_0$ and $\mathbf{y}_0$ at $k=-\Omega/2$. From $\mathbf{y}(-\Omega/2)$ and $\mathbf{y}(\Omega/2)$, we determine a constant $\varphi_0$ to enforce  (\ref{eq:realcon}) and 
the Zak phase $\varphi_{\rm zak}$ in (\ref{eq:diff}). This is followed by carrying out a gauge transform to obtain the new $\widetilde{\mathbf{y}}$ in (\ref{eq:yvectilde}) with the Berry phase $\varphi$ given by (\ref{eq:minzak}). 
Using the definition of $\widetilde{\mathbf{y}}$ in (\ref{eq:afuntilde}), we obtain a function $\widetilde{\alpha}$ in $\R$ that defines the Fourier transform of the Wannier function by (\ref{eq:wannfour}). 
Finally, we Fourier transform $\widetilde{\alpha}$ to obtain the Wannier function $W_0$ by (\ref{eq:wanift2}). 
This Wannier function is guaranteed to be optimal by Theorem
\ref{thm:wanexp} and Corollary \ref{cor:opt}, and real by Theorem
\ref{thm:real}.

We now describe a high-order numerical procedure for the construction described above. Sections \ref{sec:comfour} and \ref{sec:cominit} contain the computation for the first stage, the 
Sections \ref{s:ivp} and \ref{sec:comphase} for the second stage, and Sections \ref{sec:comwanfour} and \ref{sec:comwan} for the last stage.

 
\subsection{Numerical procedure}

The input to our method are the potential function $V$ and its period $a$ (so $\Omega = \frac{2\pi}{a})$, $M \in \mathbb{N}$ specifying
the number of degrees of freedom in the real space, 
$K \in \mathbb{N}$ specifying the number of degrees of freedom in the momentum space, and 
 $\ell \in \mathbb{N}$ indicating the band index.

\subsubsection{Computing the Fourier interpolant of the potential\label{sec:comfour}}
We first let $\{t_j\}_{j=1}^{2M+1}$ be given by
\begin{equation}
t_j = -\frac{\pi}{\Omega} + \frac{2 \pi (j-1)}{\Omega (2M+1)}, \quad j = 1,2,\hdots,2M+1.
\label{eq:tpoints}
\end{equation}
We then evaluate $V(t_j)$ for all $j = 1,2,\hdots,2M+1$. 
The $2M+1$ Fourier coefficients $\{\widehat{V}_j\}_{j=-M}^M$ in \eqref{eq:vfour}
are then computed via a discrete Fourier transform.
We then form the potential matrix $\mathbf{V} \in \mathbb{C}^{(2M+1) \times
(2M+1)}$ with entries given by 
\begin{equation}
\mathbf{V}_{m,n} = 
\begin{cases} \widehat{V}_{m-n} & |m-n| \leq M \\
0 & |m-n| > M \end{cases}, \quad m,n = 1,2,\hdots,2M+1.
\label{eq:potmatentries}
\end{equation}

\begin{remark}
Computing $\{\widehat{V}_j\}_{j=-M}^M$ can be
done in $O(M \log M)$ floating point operations using the fast Fourier transform.
Forming the matrix $\mathbf{V}$ takes $O(M^2)$ floating point operations.
\end{remark}

\subsubsection{Obtaining an initial condition \label{sec:cominit}}
We define $\mathbf{D}(k) \in \mathbb{C}^{(2M+1) \times (2M+1)}$ to have
entries 
\begin{equation}
\mathbf{D}(k)_{m,n} =  (k+(m-M-1) \Omega)^2 \delta_{mn}, \quad m,n=1,2,\hdots,2M+1.
\label{eq:dmatentries}
\end{equation}
We then solve the eigenvalue problem
\begin{equation}
\left(\mathbf{D} \left( -\frac{\Omega}{2} \right) + \mathbf{V} \right) \mathbf{y}^{(0)} = E_0 \mathbf{y}^{(0)},
\label{eq:eig}
\end{equation}
where $E_0$ is the $\ell$th smallest eigenvalue and $\mathbf{y}^{(0)}$ is the
corresponding eigenvector, normalized as in (\ref{eq:ynorm}) such that 
\begin{equation}
\| \mathbf{y}^{(0)} \| = 1/a.
\label{eq:normalize}
\end{equation}
This can be done using the standard QR algorithm (cf.\ e.g.,\ \cite[Section
8.3]{golub2013matrix}).
This procedure would then take $O(M^3)$ floating
point operations.
However, it should be noted that the coefficient matrix in \eqref{eq:eig} is
highly structured.
In particular, it consists of a diagonal matrix plus a Toeplitz matrix,
which can likely be exploited to accelerate this step.

\begin{remark}
The convergence of the discretization depends on the smoothness of $V$. 
If $V \in C^m(\mathbb{R};\mathbb{R})$, then the error will decrease at a rate of
$O(M^{-m})$. 
In particular, if $V \in C^\infty(\mathbb{R};\mathbb{R})$, then the error will
decrease superalgebraically (i.e.,\ at a rate faster than $O(M^{-m})$ for
any $m \in \mathbb{N}$). 
We refer the reader to \cite[Chapter 2]{boyd2001chebyshev} for more details.
\label{rmk:vsmooth}
\end{remark}

\subsubsection{Solving the initial value problem}
\label{s:ivp}
We now define $\mathbf{S} \in \mathbb{C}^{(2M+1) \times (2M+1)}$ to have
entries
\begin{equation}
\mathbf{S}(k)_{m,n} = 2(k+(m-M-1) \Omega)) \delta_{mn}, \quad m,n=1,2,\hdots,2M+1.
\label{eq:smatdisc}
\end{equation}
We let $\mathbf{I}$ denote the $(2M+1) \times (2M+1)$ identity matrix.
We now solve the initial value problem 
\begin{align}
 \mathbf{y}'(k) &=
-(\mathbf{D}(k) - E(k) \mathbf{I} + \mathbf{V})^\dagger \mathbf{S}(k) \mathbf{y}(k) \label{eq:discivp} \\
E'(k) &= \frac{\mathbf{y}(k)^* \mathbf{S}(k) \mathbf{y}(k)}{\mathbf{y}(k)^*\mathbf{y}(k)}
\label{eq:discenergyode} \\
\mathbf{y} \left( -\frac{\Omega}{2} \right) &= \mathbf{y}^{(0)} \\
E \left( - \frac{\Omega}{2} \right) &= E_0,
\end{align}
where $E_0$ and $\mathbf{y}^{(0)}$ were obtained in \eqref{eq:eig}.
We would like to solve this initial value problem from $-\Omega/2$ to
$\Omega/2$.
We note that by \eqref{eq:normalize}
$\|\mathbf{y}(k)\| = 1/a$ for $k \in [-\Omega/2, \Omega/2]$. 

We use a time-stepping scheme to determine the values
\begin{equation}
\mathbf{y}^{(j)} \approx \mathbf{y}(k_j), \quad j =1,2,\hdots,K,
\end{equation}
where
\begin{equation}
k_j = -\frac{\Omega}{2} + \frac{\Omega (j-1)}{K-1}, \quad j = 1,2,\hdots,K.
\label{eq:kj}
\end{equation}
At each time-step the coefficient matrix
\begin{equation}
\mathbf{\Theta}(k) = \mathbf{D}(k) + \mathbf{V} - E(k) \mathbf{I} 
\label{eq:thetamatdisc}
\end{equation}
must be inverted. 
In this paper, we compute $\mathbf{\Theta}(k)^\dagger$ by 
first computing a singular value decomposition of
$\mathbf{\Theta}(k)$, discarding the mode corresponding to the smallest
singular value, and applying the inverse of the decomposition directly to
the right-hand side (cf.\ e.g.,\ \cite[Section 2.6.1]{bjorck2015numerical}). 
It is not as efficient as the approach in Remark \ref{rmk:inv}, but this is not
significant for one-dimensional problems.

Since this problem is not stiff, any time-stepping scheme may be chosen.
In our numerical experiments, we use the standard 4th-order Runge--Kutta
method to solve the initial value problem for simplicity (cf.\ e.g.,\
\cite[Section 12.5]{suli2003introduction}); 
the resulting solution produces a solution that converges at rate
$O(K^{-4})$.
Even higher order, for example $O(K^{-12})$, can be achieved by switching to a spectral
deferred scheme (cf.\ \cite{dutt2000spectral}). 
At each time-step, the dominant cost is that of computing the singular value
decomposition of $\mathbf{\Theta}(k)$, which consists of $O(M^3)$ floating
point operations. 
Since $K$ time steps are taken, this stage of our procedure requires
$O(KM^3)$ floating point operations. 

We note that the explicit dependence of \eqref{eq:discivp} on $E(k)$ and the ODE for the energy (\ref{eq:discenergyode}) can be eliminated altogether by replacing $E(k)$ with the Rayleigh quotient 
\begin{equation}
\frac{\mathbf{y}(k)^*
(\mathbf{D}(k) + \mathbf{V}) \mathbf{y}(k)}{\mathbf{y}(k)^* \mathbf{y}(k)}
\end{equation}
when $\mathbf{y}(k)$ is known. This also has the advantage that the error the computed $E(k)$ is the squared error in the eigenvectors from the ODE solve.

 \begin{remark}
 	If two vectors $\mathbf{v}_1=\mathbf{y}(k)$ and $\mathbf{v}_2=\mathbf{y}(k+\Delta k)$	for small $\Delta k$ are obtained by direct eigensolves, the phases of $\mathbf{v}_1$ and $\mathbf{v}_2$ are independent. In \cite{marzari1997maximally, marzari2012maximally, vanderbilt2018berry}, it is pointed out that  one can approximately parallel transport to $\widetilde{\mathbf{v}}_2$ from $\mathbf{v}_1$ given by 
 $		\widetilde{\mathbf{v}}_2 = e^{\I \beta}\mathbf{v}_2$, 
 where $\beta$ is the phase given by $\beta = -\Im\log(\mathbf{v}_1^*\mathbf{v}_2)$. 
However, the convergence of this procedure is fundamentally limited to $O((\Delta
 k)^2)$,  {so an optimization step is required in order to achieve the global optimality.     As mentioned in the introduction, a corollary of the analysis in Section\,\ref{sec:opt} is that   the Wannier functions produced by this approximate parallel transport scheme remain exponentially localized, though they differ from those obtained by the method presented in this paper by $O(\Delta k^2)$\,. A detailed analysis will be reported at a later date.
 }
 \label{rmk:twist}
 \end{remark}

\subsubsection{Correcting the phase\label{sec:comphase}}
According to Lemma \ref{lem:zak}, we can determine $\varphi_{\rm zak}$ from  
\begin{equation}
  \mathbf{y}^{(K)}_m = e^{\I \varphi_{\rm zak}}  \mathbf{y}^{(0)}_{m+1}, 
\end{equation}
for any $m \in \{1,2,\hdots,2M \}$ to fix the jumps.
Moreover, the condition in (\ref{eq:realcon}) allows us to determine $\varphi_0$ by
\begin{equation}
	e^{-\I \varphi_{0}}\mathbf{y}^{(0)}_{m} = e^{\I \varphi_{0}}\overline{\mathbf{y}}^{(K)}_{2M+2-m}
\end{equation}
for any $m \in \{M+1,M+2,\hdots,2M+1\}$ for the realty of Wannier functions. By (\ref{eq:minzak2}), we define the Berry phase
\begin{equation}
	\varphi(k_j) = \frac{\varphi_{\rm zak}}{\Omega} k_j,\quad j=1,2,\ldots, K.
		\label{eq:phik}
\end{equation}
We form the new vector 
\begin{equation}
	\widetilde{\mathbf{y}}^{(j)} = e^{-\I\varphi_0}e^{-\I \varphi(k_j)}\mathbf{y}^{(j)},\quad j=1,2,\ldots, K,
	\label{eq:phi0}
\end{equation}
which will lead to the real optimally localized Wannier function.



\subsubsection{Computing $\widehat{W}_0$\label{sec:comwanfour}}
Given $x \in \mathbb{R}$, we then evaluate $W_0(x)$ using a trapezoidal rule
discretization of \eqref{eq:wanift2}.
We set 
\begin{equation}
N = (2M+1)K 
\end{equation}
and define $\{ \widetilde{\alpha}_l \}_{l=1}^N$ by 
\begin{equation}
\widetilde{\alpha}_{(m-1)K+j} = \widetilde{\mathbf{y}}^{(j)}_m, \quad j = 1,2,\hdots,K,\, m = 1,2,\hdots,2M+1.
\label{eq:alphaform}
\end{equation}

\subsubsection{Evaluate the Wannier function \label{sec:comwan}}
Finally, we compute 
\begin{equation}
W_0(x) = \frac{1}{K} \sum_{j=1}^{K} \sum_{m=1}^{2M+1}  \widetilde{\alpha}_{(m-1)K+j} e^{\I x (k_j+ (m-M-1) \Omega)}
\label{eq:trap}
\end{equation}

The rate of convergence of the sum in \eqref{eq:trap} with respect to $K$ and
$M$ depends on the error in the trapezoidal rule approximation to the integral in \eqref{eq:wanift2}. 
The integrand in \eqref{eq:wanift2} is analytic and decays as $\xi \to \pm \infty$. 
Thus, it follows from the Euler--Maclaurin formula that the precise rate of convergence will depend on the rate at which the integrand goes to zero, which  in turn is determined by the smoothness of the potential (see Remark\,\ref{rmk:vsmooth}). 
When the integrand decays exponentially $\xi \to \pm \infty$, the error in the trapezoidal rule approximation decays with respect to $K$ as $O(e^{- c K})$ for some $c > 0$ respect to $K$ as $O(e^{- c K})$ for some $c > 0$ (cf.\ e.g.,\ \cite{trefethen2014exponentially}). 

We note that while the cost of direct evaluation of \eqref{eq:trap} for each
$x$ is $O(KM)$,
the evaluation of \eqref{eq:trap} for a collection of points can be
accelerated using the fast Fourier transform or non-uniform fast Fourier
transform (cf.\ e.g.,\ \cite{greengard2004accelerating}).

\section{Detailed description of the procedure}
\label{s:detailed}
The input of our numerical procedure is the lattice spacing $a > 0$,
potential function $V$, $M \in \mathbb{N}$ specifying the number of
spatial degrees of freedom, $K \in \mathbb{N}$ specifying the number of
degrees of freedom in momentum space, $\ell \in \mathbb{N}$
indicating the band index, and a collection of points $\{x_j\}_{j=1}^P$ on
which the Wannier function is to be evaluated. 
The output of Algorithm 1 is $\{W_0(x_j)\}_{j=1}^P$.

For clarity, we have isolated the evaluation of the right-hand side of the
ODEs \eqref{eq:discivp} and \eqref{eq:discenergyode} into Algorithm 1A. 
Algorithm 1A takes as input $\mathbf{y}(k) \in \mathbb{C}^{2M+1}$, $E(k) \in
\mathbb{R}$, and $k \in \mathbb{R}$ and outputs the right-hand side of the
ODE in \eqref{eq:discivp} and \eqref{eq:discenergyode}.

\noindent \underline{\textbf{Algorithm 1}} 

\noindent \textit{Input:} $a > 0$, $M \in \mathbb{N}$, $K \in
\mathbb{N}$, $V : \mathbb{R} \to \mathbb{R}$, $\ell \in \mathbb{N}$,
$\{x_j\}_{j=1}^P \subset \mathbb{R}$

\noindent \textit{Output:} $\{W_0(x_j)\}_{j=1}^P$ 

\noindent \textbf{Step 1} [Forming the potential matrix $\mathbf{V}$]. 

\begin{enumerate}
\item Evaluate $\{t_j\}_{j=1}^{2M+1}$ by \eqref{eq:tpoints}.

\item Evaluate $\{V(t_j)\}_{j=1}^{2M+1}$.

\item Apply the discrete Fourier transform to $\{V(t_j)\}_{j=1}^{2M+1}$ to produce
the coefficients of the Fourier interpolant $\{\widehat{V}_j\}_{j=-M}^M$.

\item Form $\mathbf{V} \in \mathbb{C}^{(2M+1) \times (2M+1)}$ using the
formula \eqref{eq:potmatentries}.
\end{enumerate}

\noindent \textbf{Step 2} [Obtaining the initial condition]. 

\begin{enumerate}

\item Compute the matrix $\mathbf{D}(-\Omega/2)$ using the formula \eqref{eq:dmatentries}.

\item Solve the eigenvalue problem \eqref{eq:eig} for the $\ell$th smallest
eigenvalue $E^{(0)}$ and corresponding eigenvector $\mathbf{y}^{(0)}$ using
the QR algorithm and ensuring $\| \mathbf{y}^{(0)} \| = 1$.
\end{enumerate}

\noindent \textbf{Step 3} [Solving the initial value problem using RK4].

\begin{enumerate}

\item Set $k_j$ for $j = 1,2,\hdots,K$ via \eqref{eq:kj}. Set $\Delta k = \Omega / K$. 

\item Do $j = 1, 2, \hdots, K$

\begin{enumerate}

\item Use Algorithm 1A with inputs $\mathbf{V}$, $\mathbf{y}^{(j-1)}$, and $k_j$ to obtain $\mathbf{w}^{(1)}$ and $F^{(1)}$.

\item Use Algorithm 1A with inputs $\mathbf{V}$, $\mathbf{y}^{(j-1)} +
(\Delta k/2) \mathbf{w}^{(1)}$, and $k_j + \Delta k/2$ to obtain
$\mathbf{w}^{(2)}$ and $F^{(2)}$.

\item Use Algorithm 1A with inputs $\mathbf{V}$, $\mathbf{y}^{(j-1)} +
(\Delta k/2) \mathbf{w}^{(2)}$, and $k_j + \Delta k/2$ to obtain
$\mathbf{w}^{(3)}$ and $F^{(3)}$.

\item Use Algorithm 1A with inputs $\mathbf{V}$, $\mathbf{y}^{(j-1)} +
(\Delta k) \mathbf{w}^{(3)}$, and $k_j + \Delta k$ to obtain
$\mathbf{w}^{(4)}$ and $F^{(4)}$.

\item Set 
\begin{align}
\mathbf{y}^{(j)} &= \mathbf{y}^{(j-1)} + \frac{\Delta k}{6} \left( \mathbf{w}^{(1)}+ 2\mathbf{w}^{(2)} + 2\mathbf{w}^{(3)} + \mathbf{w}^{(4)}\right) \label{eq:rkstep} \\
E^{(j)} &= E^{(j-1)} + \frac{\Delta k}{6} \left( F^{(1)} + 2F^{(2)} + 2F^{(3)} + F^{(4)}\right).
\end{align}

\end{enumerate}

\item[] End Do

\end{enumerate}

\noindent \textbf{Step 4} [Phase correction]

\begin{enumerate}
\item Set $\varphi_0$ by (\ref{eq:phi0}).
\item Set $\varphi_{\rm zak} = \arg (\mathbf{y}_m^{(K)} / \mathbf{y}_{m+1}^{(0)})$.

\item Do $j=1,2,\hdots,K$ 

\begin{enumerate}
\item Set $\widetilde{\mathbf{y}}^{(j)} = \exp(-\I \varphi_0)\exp \left( -\I \varphi_{\rm zak}  k_j/\Omega \right) \mathbf{y}^{(j)}$.
\end{enumerate}

\item[] End Do

\end{enumerate}

\noindent \textbf{Step 5} [Tabulate Fourier transform of $W_0$] 

\begin{enumerate}
\item Set $N = (2M+1)K$.

%
%

\item Compute $\{ \widetilde{\alpha}_l\}_{l=1}^N$ by \eqref{eq:alphaform}.
\end{enumerate}

\noindent \textbf{Step 6} [Evaluate Wannier function]

\begin{enumerate}
\item Compute $\{ W(x_j) \}_{j=1}^P$ using \eqref{eq:trap}.
\end{enumerate}

\noindent \underline{\textbf{Algorithm 1A}} 

\noindent \textit{Input:} $\mathbf{V} \in \mathbb{C}^{(2M+1) \times (2M+1)}$, $\mathbf{y}(k) \in \mathbb{C}^{2M+1}$, $E(k) \in \mathbb{R}$, $k \in \mathbb{R}$

\noindent \textit{Output:} $\mathbf{y}'(k) \in \mathbb{C}^{2M+1}$, $E'(k) \in \mathbb{R}$

\begin{enumerate}
\item Evaluate $\mathbf{S}(k)$ using \eqref{eq:smatdisc}.
\item Compute $E'(k)$ using \eqref{eq:discenergyode}.
\item Compute $\mathbf{z} = \mathbf{S}(k) \mathbf{y}(k)$.
\item Evaluate $\mathbf{\Theta}(k)$ defined in \eqref{eq:thetamatdisc} using \eqref{eq:dmat}.
\item Compute the singular value decomposition of $\mathbf{\Theta}(k) = \mathbf{U} \mathbf{\Sigma} \mathbf{V}^*$ using the QR algorithm.
\item Set $\mathbf{y}'(k) = \mathbf{V}_{1:2M+1,1:2M} (\mathbf{\Sigma}_{1:2M,1:2M})^{-1} (\mathbf{U}_{1:2M+1,1:2M})^* \mathbf{z}$.
\end{enumerate}

\section{Numerical experiments}
\label{s:numerics}

We implemented the procedure described in Section \ref{s:procedure} in
Fortran.
We fix $V$, $a$, $M$, and $\ell$ and increase $K$ until 10 digits of accuracy
is obtained.
The timings reported are the CPU times required to obtain the Fourier
transform of $W_0$ evaluated on a uniform grid (i.e., Steps 1--5 in Section
\ref{s:detailed}). 
To obtain an estimate of the error we compute
\begin{equation}
E_{\rm RK4}^{(\ell)} =  \left\| \mathbf{y}^{(K)} - \mathbf{y} \left(\frac{\Omega}{2} \right) \right\|, 
\end{equation}
where $\mathbf{y}^{(K)}$ is defined in \eqref{eq:rkstep} and
$\mathbf{y}(\Omega/2)$ is computed by directly solving the eigenvalue problem
via the QR algorithm.
We also report 
\begin{equation}
E_{\rm imag}^{(\ell)} = \frac{\max_{j=1,2,\hdots,1000} \left| \operatorname{Im} \left[ W_0 \left(
-\frac{\pi}{\Omega} + \frac{2 \pi j}{1000 \Omega} \right) \right] \right|}
{\max_{j=1,2,\hdots,1000} \left|  W_0 \left(
-\frac{\pi}{\Omega} + \frac{2 \pi j}{1000 \Omega} \right) \right|},
\end{equation}
where we evaluate $W_0$ using \eqref{eq:trap}.
All experiments were conducted in double precision on a single core of an
Apple M3 PRO CPU with the GNU Fortran compiler.

\subsection{Gaussian potential}

\begin{figure}[t]
\centering
\includegraphics[scale=0.45]{\figpath 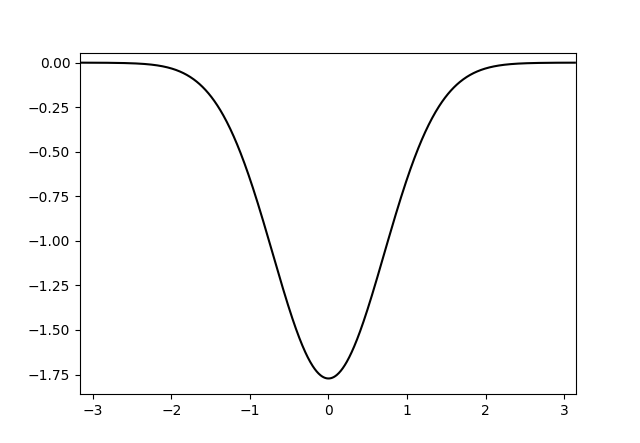}
\includegraphics[scale=0.42]{\figpath 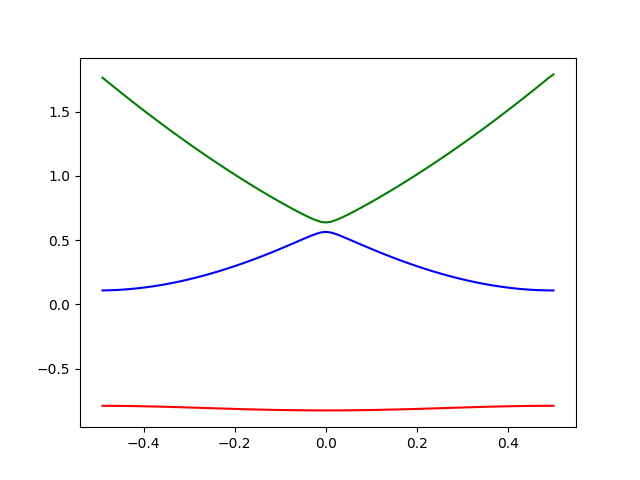}
\vspace{-2em}
\caption{Plot of potential (left) and $E^{(j)}(k)$ (right) for $j=1,2,3$ (bottom,
middle, top, resp.) }
\label{fig:gauss}
\end{figure}

In this experiment, we set $a = 2\pi$ (so $\Omega=1$), $M=10$, and $\ell = 1,2,3$.
The potential is taken to be
\begin{equation}
V(x) = - \frac{1}{2} - \sum_{j=1}^{5} e^{- j^2 / 4} \cos(j \Omega x)  
\end{equation}
and plotted in Figure \ref{fig:gauss} (left). 
The corresponding energies are plotted in Figure \ref{fig:gauss} (right). 
The Wannier functions are plotted in Figure \ref{fig:gaussian_wannier}.

\begin{figure}[t!]
\begin{tabular}{cc}
\includegraphics[scale=0.4]{\figpath 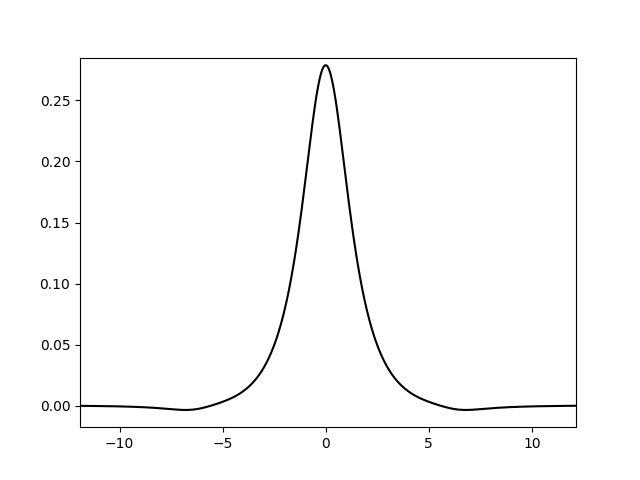} &
\includegraphics[scale=0.4]{\figpath 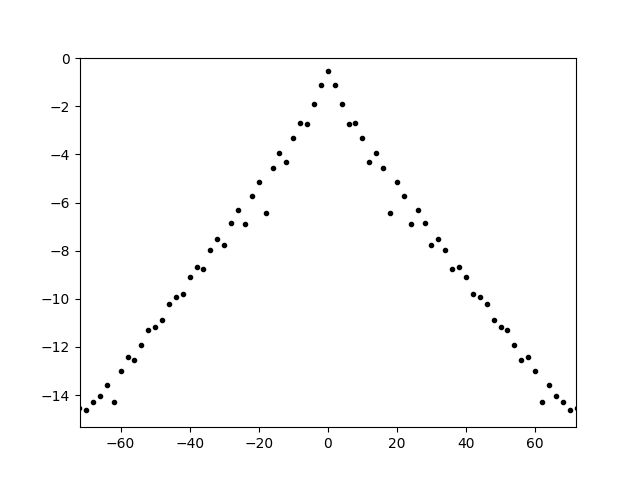} \\
\includegraphics[scale=0.4]{\figpath 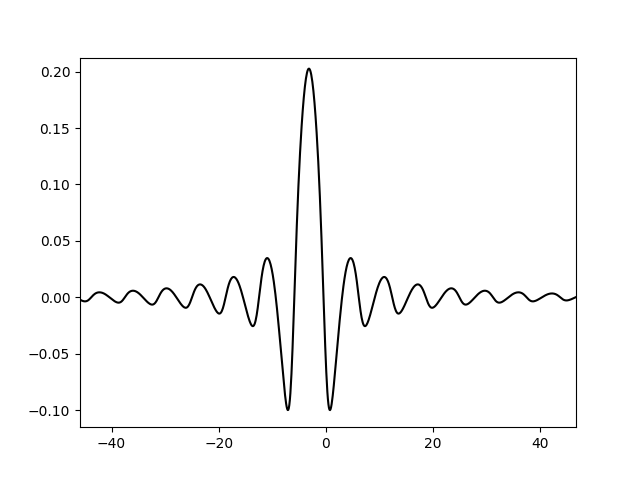} &
\includegraphics[scale=0.4]{\figpath 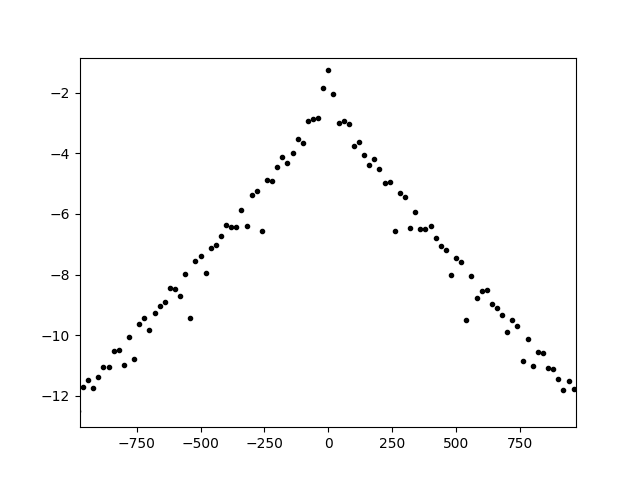} \\
\includegraphics[scale=0.4]{\figpath 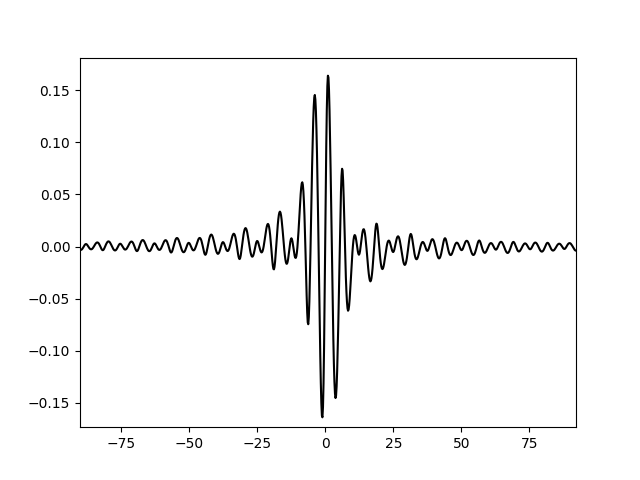} &
\includegraphics[scale=0.4]{\figpath 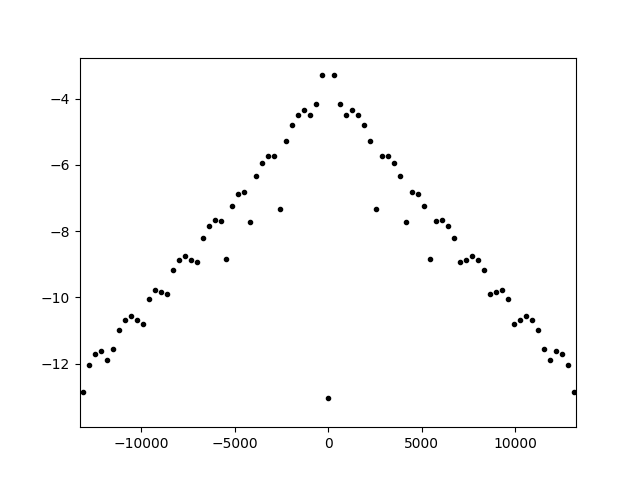} \\
\end{tabular}
\caption{$W_0^{(j)}(x)$ (left) and $\log_{10}(| W_0^{(j)}(x)|)$ (right) for $j=1,2,3$ (top,
middle, bottom, resp.)
We note that while all the Wannier functions decay exponentially, the constants in this rate depend strongly on the band, as demonstrated by the different horizontal axes scales.
	}
\label{fig:gaussian_wannier}
\end{figure}
The relevant quantities are: 
\begin{center}
\begin{tabular}{c|c|cc|cc|cc}
$K$ & time (s) & $E_{\rm RK4}^{(1)}$ & $E_{\rm imag}^{(1)}$  
& $E_{\rm RK4}^{(2)}$ & $E_{\rm imag}^{(2)}$  
& $E_{\rm RK4}^{(3)}$ & $E_{\rm imag}^{(3)}$  
\\
\hline
51      &  2.11E--01 & 2.04E--09 & 5.28E--10 & 1.66E--04 & 3.81E--04  & 7.17E+00 & 3.53E--01 \\
101     &  3.42E--01 & 1.33E--10 & 3.26E--11 & 8.18E--06 & 1.36E--05  & 6.53E+00 & 1.20E--01 \\
201     &  6.66E--01 & 8.51E--12 & 2.03E--12 & 5.58E--07 & 9.56E--07  & 2.66E+00 & 1.56E--02 \\
401     &  1.25E+00  &  &  & 3.60E--08 & 6.30E--08  & 4.35E--01 & 7.57E--04 \\
801     &  2.42E+00  &  &  & 2.28E--09 & 4.02E--09  & 3.51E--02 & 2.36E--05 \\
1601    &  4.82E+00  &           &           & 1.43E--10 & 2.54E--10  & 1.56E--03 & 1.62E--06 \\
3201    &  9.39E+00  &           &           & 8.94E--12 & 1.59E--11  & 9.17E--05 & 1.27E--07 \\
6401    &  1.84E+01  &           &           &  &    & 5.59E--06 & 8.67E--09 \\
12801   &  3.62E+01  &           & 	     &   &   & 3.44E--07 & 5.64E--10 \\
25601   &  7.06E+01  &       	 &  	     &           &            & 2.22E--08 & 3.78E--11 \\
51201   &  1.38E+02  &      	 & 	     &           &            & 2.70E--11 & 2.57E--12 \\
\end{tabular}
\end{center} 

We observe that the expected fourth-order convergence with respect to $K$
down to an accuracy of about 10 digits.
Moreover, the $O(K)$ complexity can be observed.

\subsection{Asymmetric potential}
\begin{figure}[t!]
\centering
\includegraphics[scale=0.47]{\figpath 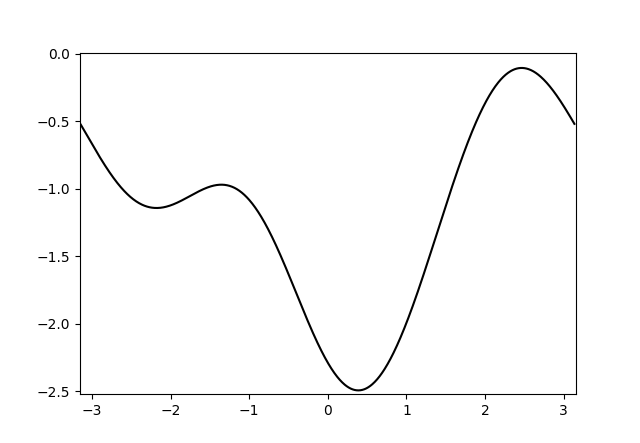}
\includegraphics[scale=0.43]{\figpath 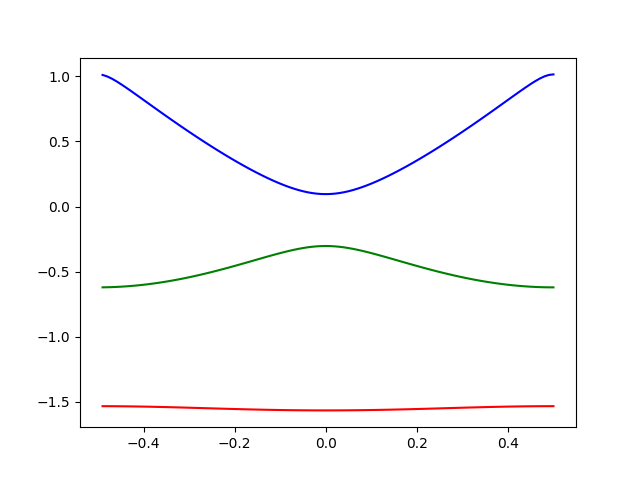}
\vspace{-2em}
\caption{Plot of potential (left) and $E^{(j)}(k)$ (right) for $j=1,2,3$ (bottom,
middle, top, resp.) }
\label{fig:asymmetric}
\end{figure}

In this experiment, we set $a = 2 \pi$ (so that $\Omega=1$), $M=15$, and $\ell = 1,2,3$.
The potential is given by 
\begin{equation}
V(x) = -\frac{1}{4} \left( 1 + 2 \sin(2 \Omega x) + 3 e^{\cos(\Omega x)}\right).
\end{equation}
The potential is plotted in Figure \ref{fig:asymmetric} (left) and the energies are plotted in Figure \ref{fig:asymmetric} (right).
The Wannier functions are plotted in Figure \ref{fig:asymmetric_wannier}.
The relevant quantities are: 

\begin{center}
\begin{tabular}{c|c|cc|cc|cc}
$K$ & time (s) & $E_{\rm RK4}^{(1)}$ & $E_{\rm imag}^{(1)}$  
& $E_{\rm RK4}^{(2)}$ & $E_{\rm imag}^{(2)}$  
& $E_{\rm RK4}^{(3)}$ & $E_{\rm imag}^{(3)}$  
\\
\hline
51      & 1.00E--01 & 2.79E--09 & 7.18E--10 & 9.63E--07 & 5.08E--07 & 1.16E--02 & 2.68E--05 \\
101     & 1.68E--01 & 1.78E--10 & 4.45E--11 & 6.20E--08 & 3.22E--08 & 6.03E--04 & 2.00E--06 \\
201     & 2.90E--01 & 1.12E--11 & 2.77E--12 & 3.92E--09 & 2.02E--09 & 3.57E--05 & 1.43E--07 \\
401     & 5.36E--01 &  &  & 2.47E--10 & 1.27E--10 & 2.16E--06 & 9.46E--09 \\
801     & 1.02E+00  & & & 1.57E--11 & 8.07E--12 & 1.33E--07 & 6.08E--10 \\
1601    & 2.01E+00  & 	 & 		    &  &  & 8.24E--09 & 3.86E--11 \\ 
3201    & 4.00E+01  & 	 & 		    & & & 5.07E--10 & 2.47E--12 \\ 
6401    & 7.89E+01  & 	 & 		    &		&	    & 6.49E--11 & 4.07E--13 \\ 
\end{tabular}
\end{center}
We are again able to rapidly compute a solution down to 10 digits of accuracy. 

\begin{figure}[t!]
\begin{tabular}{cc}
\includegraphics[scale=0.4]{\figpath 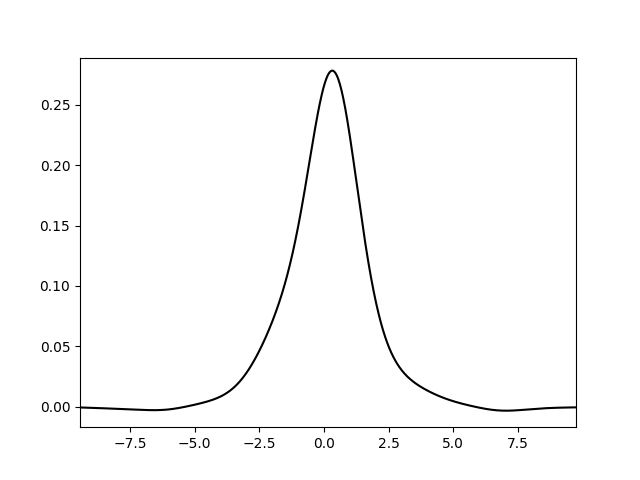} &
\includegraphics[scale=0.4]{\figpath 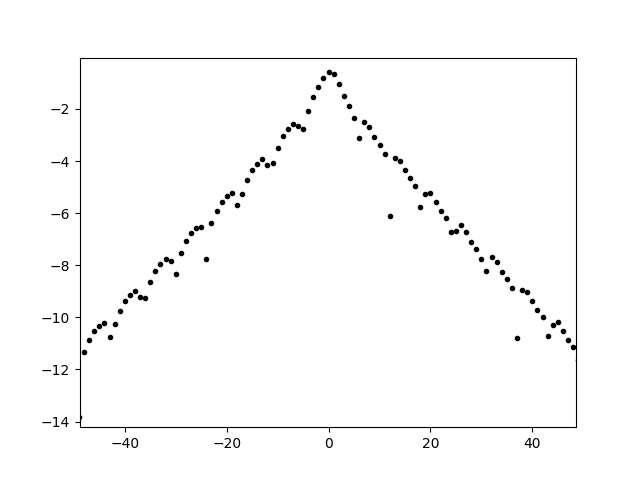} \\
\includegraphics[scale=0.4]{\figpath 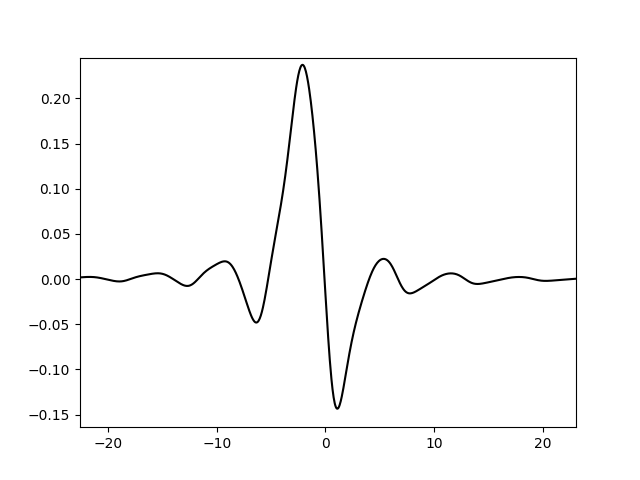} &
\includegraphics[scale=0.4]{\figpath 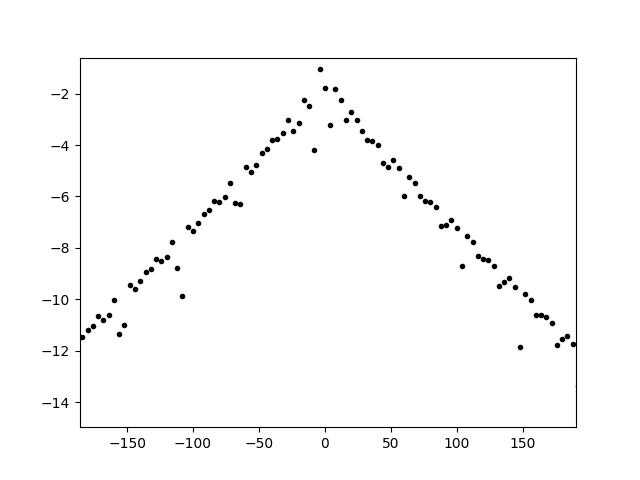} \\
\includegraphics[scale=0.4]{\figpath 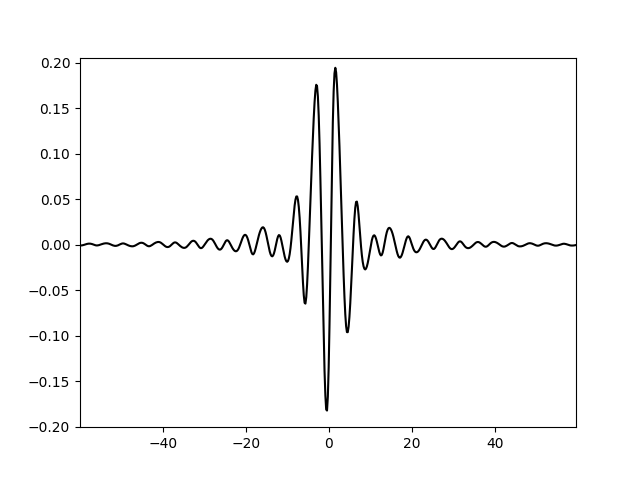} &
\includegraphics[scale=0.4]{\figpath 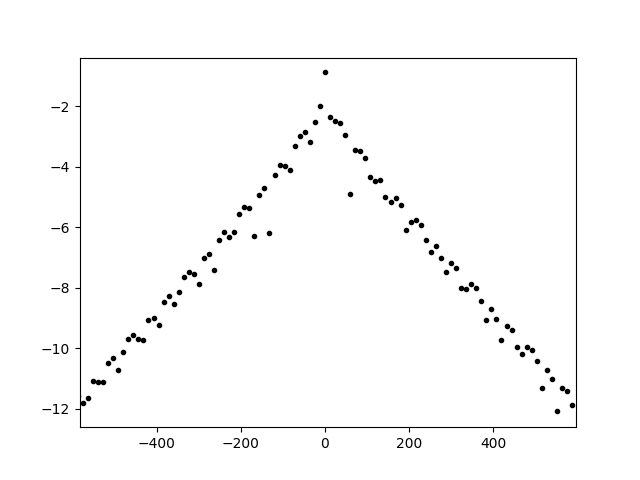} \\
\end{tabular}
\caption{$W_0^{(j)}(x)$ (left) and $\log_{10}(| W_0^{(j)}(x)|)$ (right) for $j=1,2,3$ (top,
middle, bottom, resp.)}
Again, the different horizontal axes scales show that the constants in the decay depend on the band.
\label{fig:asymmetric_wannier}
\end{figure}

\section{Conclusion and generalization}
\label{s:conclusions}

In this work, we developed a robust, high-order numerical method for
computing Wannier functions for one-dimensional crystalline insulators. We presented complete analysis that shows the constructed Wannier functions are exponentially localized and globally optimal in terms of their variance. Furthermore, the Wannier functions can always be chosen to be real. The analysis and algorithm in this paper can also be viewed as a constructive proof of the existence of exponentially localized Wannier functions for an isolated single band in one dimension. We note that the construction can be applied easily to many other self-adjoint analytic families of operators beyond the Schr\"odinger operators in (\ref{eq:schro}).

 { 
The key challenge in extending the scheme presented here to the multiband case is that, instead of scalar phases as in the single-band case in this paper, all possible gauge choices form a family of unitary matrices, i.e. the matrix $U(N_b)$ group where $N_{b}$ is the number of bands considered. While this makes the analysis somewhat more involved, the underlying concepts remain largely the same, and globally optimal solutions can still be computed analytically.} 
{ In higher dimensions, topological phenomena emerge in the construction of Wannier functions, requiring ideas beyond those presented in this paper, as demonstrated in \cite{zhang2025constructing} for matrix models. Nonetheless, the one-dimensional algorithms serve as the fundamental building blocks for higher-dimensional cases. Extensions to higher dimensions have been worked out and are currently being prepared for publication.} \\[2em]

%

\noindent {\bf Acknowledgements} \\
The authors are grateful to John Schotland for drawing their attention to this subject, and to Vladimir Rokhlin for valuable discussions pertaining to this work. The authors also thank the anonymous referees for their helpful suggestions.

\newpage

\bibliographystyle{plain}
\bibliography{refs} 

\end{document}